\documentclass[10pt,a4paper]{article}
\usepackage{hyperref}
\usepackage{graphicx,subfig}
\usepackage{amsmath}
\usepackage{amssymb}
\usepackage{booktabs}
\usepackage{algorithmic}
\usepackage{algorithm}
\usepackage{xcolor}
\usepackage{enumerate}
\usepackage{tabularx}
\usepackage{psfrag}
\usepackage{xcolor}
\usepackage{amsfonts}
\usepackage{booktabs}
\usepackage{multirow}
\usepackage{amsthm}
\usepackage{cite}		

\graphicspath{{fig/}}
\definecolor{darkgreen}{rgb}{0.0, 0.4, 0.13}
\definecolor{rouge}{rgb}{0.7294,0.0392,0.0392}
\definecolor{bleu}{rgb}{0.2,0.2,0.7020}
\definecolor{bclair}{rgb}{0.6235    0.6235    0.8980}

\theoremstyle{theorem}
\newtheorem{theorem}{Theorem}

\theoremstyle{definition}
\newtheorem{definition}{Definition}
\newtheorem{example}{Example}
\newtheorem{problem}{Problem}

\theoremstyle{remark}
\newtheorem{remark}{Remark}

\date{May 2020}

\title{Parameter-free and fast nonlinear piecewise filtering.\\
Application to experimental physics.\thanks{Work supported by Defi Imag'in SIROCCO and by ANR-16-CE33-0020 MultiFracs, France.}}

\author{Barbara Pascal\thanks{Univ Lyon, ENS de Lyon, Univ Lyon 1, CNRS, Laboratoire de Physique, F-69342 Lyon, France (\texttt{firstname.lastname@ens-lyon.fr}).} \and Nelly Pustelnik\footnotemark[2] \and Patrice Abry\footnotemark[2] \and Jean-Christophe G{\'e}minard\footnotemark[2] \and Val{\'e}rie Vidal\footnotemark[2]}

\begin{document}
\maketitle

\section{Introduction}

Signals or images collected from numerous experiments in physics can be, at least as a first order approximation, described as piecewise homogeneous (piecewise constant, piecewise linear,\ldots). 
Detecting and estimating such piecewise homogeneous regions thus constitute a crucial goal to extract the physically relevant information conveyed in such data. 
This remains, however, often challenging, as signals or images are usually altered by superimposed noises, possibly with low signal-to-noise ratio, which may hinder the interpretation of the corresponding experiments.
The joint need to denoise data while preserving edges and discontinuities pertaining phase changes or region boundaries often preclude the use of classical linear filtering and call for the use of advanced nonlinear signal and image processing techniques.

Solid friction provides us with a first representative example.
Indeed, when two, nominally flat, solid surfaces in contact are forced to slide against one another,
the shear force at the contact surface exhibits generally a characteristic tooth-shape signal (Fig.~\ref{fig:phyexp}a):
the force signal thus consists of  an alternation of slow linear rises, corresponding to the loading of elastic
energy in the driving system while the surfaces in contact do not move with respect to each another,
followed by sudden drops, corresponding to fast energy releases when surfaces slide \cite{Baumberger06}.
When solids are strongly pressed one against the other,
these two {\it phases} (rest and sliding at the contact surface) can easily be identified.
However, high confinement pressures tend to damage surfaces, a major limitations in the study of the microscopic mechanisms at play. 
Therefore, probing effectively and accurately frictional material properties required that experiments are performed at low confinement pressure.
This, however, induces that collected signals have low to very low signal-to-noise-ratios~\cite{Colas19} thus requiring advanced non linear filtering signal processing techniques to detect and analyse the piecewise linear shape of data. 

\begin{figure}[t!]
\centering
\subfloat[Stick-slip: piecewise linear signal]{\includegraphics[height=4cm]{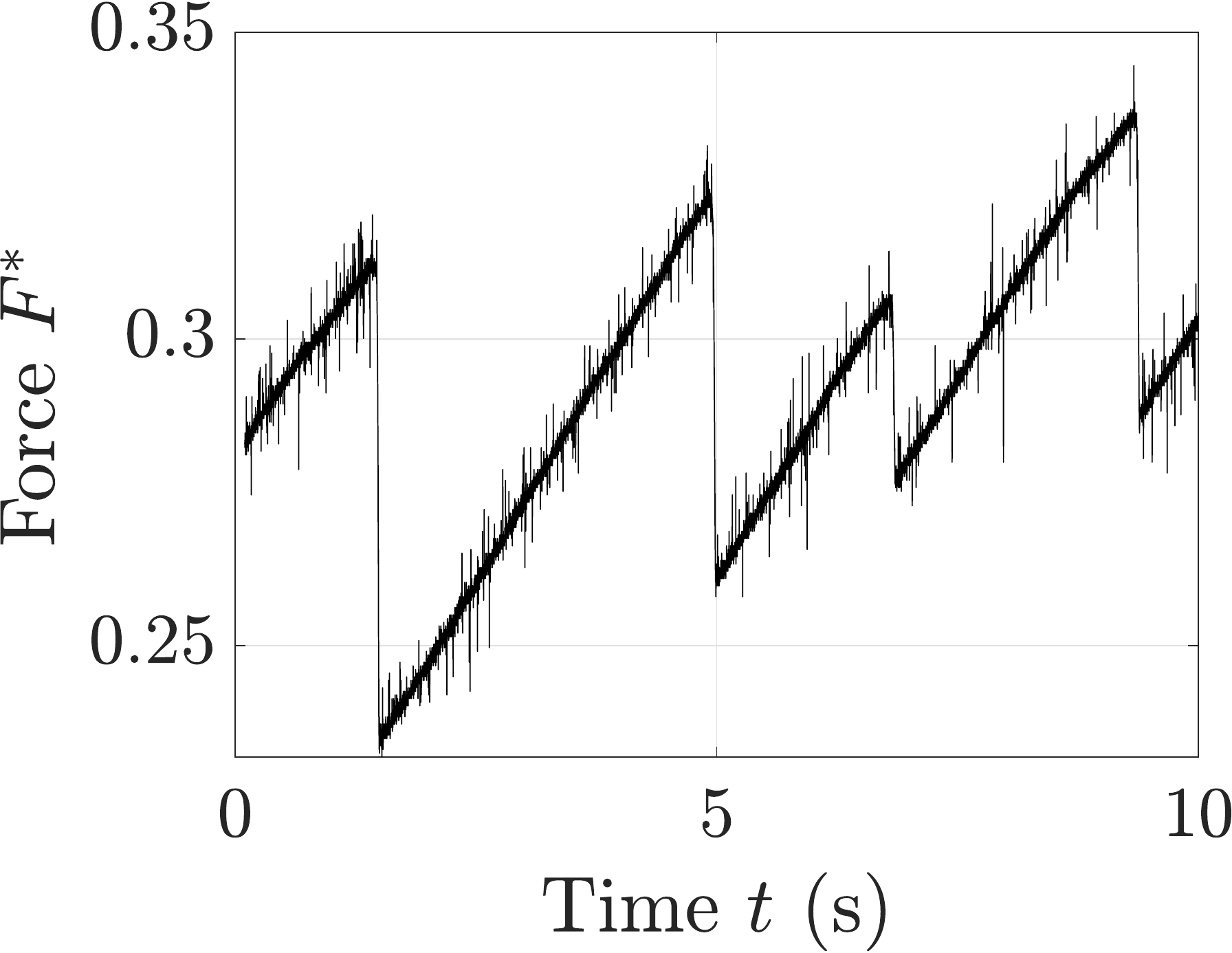}}\hspace{8mm}
\subfloat[Multiphase flow: piecewise homogeneous texture]{\raisebox{2cm}{\parbox{4.5cm}{\centering \includegraphics[height=4cm]{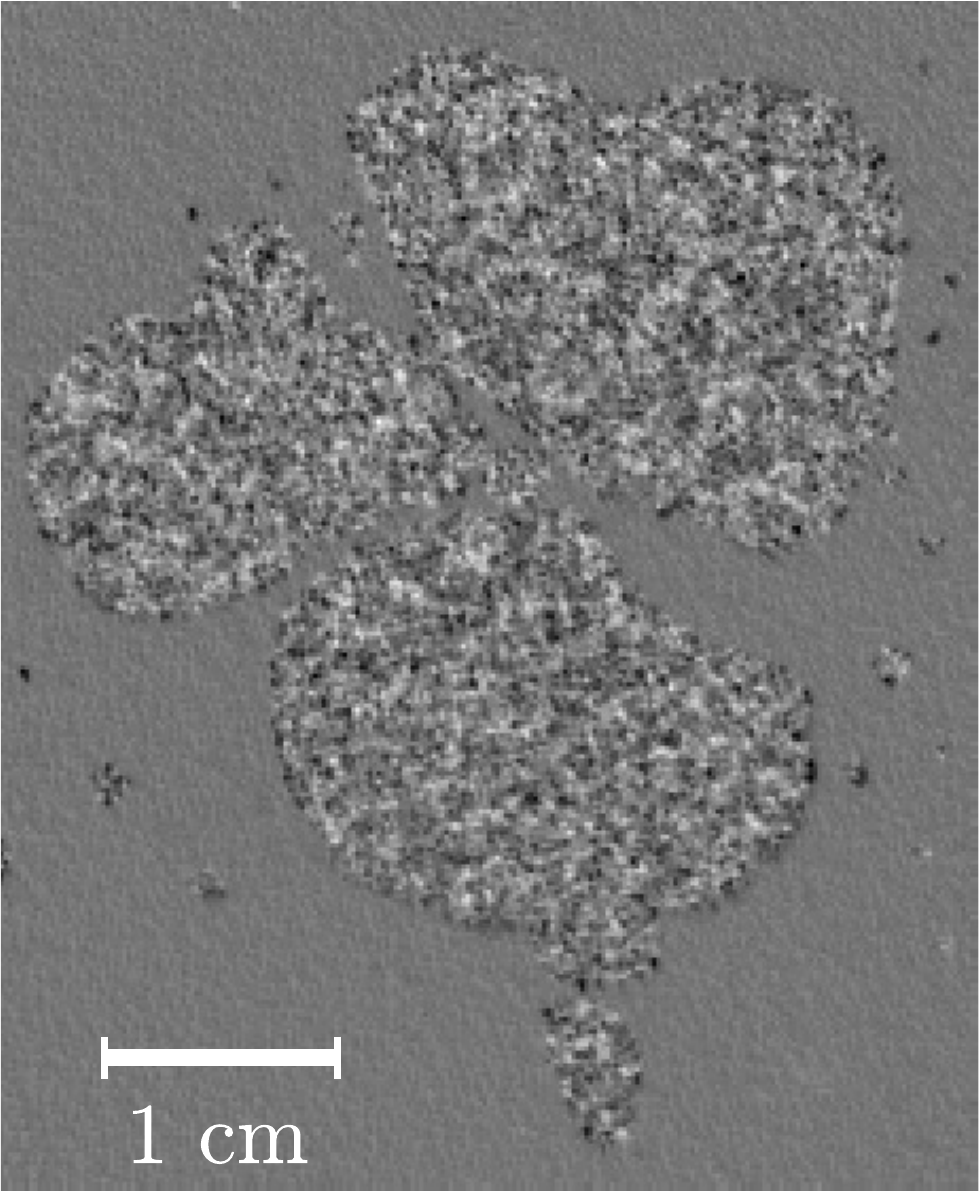}}}}
\caption{{\bf Experimental data in non linear physics.} (a): Normalized driving force $F^* = k\Delta x /(mg)$ [see Section~\ref{subsec:exp_sig}] as function of time $t$ in a solid friction experiment. (b): Direct image of a gas bubble in porous media multiphase flows (gas corresponds to the most scrambled region whereas liquid corresponds to smoother regions).\label{fig:phyexp}}
\end{figure}

Multiphase flows in porous media constitute another rich example. 
Hydrodynamics in porous media, notably mass transfer, is of prominent practical relevance in industry, e.g., for catalytic process studies. 
Hydrodynamics and mass transfer studies of multiphase flows in porous media traditionally involve packed beds and are well characterized.
However, recent experiments \cite{Serres_M_2018_j-ijmf_phe_bcc,pascal2018joint} consider innovative materials such as open-cell solid foams
which, due to high porosity and the resulting low pressure loss, are promising for industrial applications.
In such experiments, a liquid and a gas are forced to flow simultaneously through the foam and the characterization of such multiphase spatiotemporal dynamics stems from image analysis (Fig.~\ref{fig:phyexp}b). 
The challenge is here to identify liquid flows from gas bubbles. 
The rationale is that each phase can be associated to homogeneous textures in images and a crucial stake consists in identifying precisely gas bubble contours so as to measure their lengths.
Liquid and gas are both transparent and the foam is itself introducing a scrambled background, yielding low-contrast and blurred images, thus 
requiring advanced nonlinear image processing techniques form texture segmentation and contour estimation. 

These two emblematic examples share in common that the key aspects of the physics to be understood are driven by piecewise homogeneous phases. 
On one hand, studying friction requires identifying the stick and the slip phases, each associated with a piecewise linear signal. 
On other hand, studying multiphase flows implies detecting fluid phases, each associated with a piecewise homogeneous texture.
Piecewise-homogeneous signals or images are very common in numerous fields of nonlinear physics, very different in nature, such as time reversal of the magnetic field in turbulent dynamo \cite{Berhanu07}, on-off intermittency in creeping granular matter \cite{Divoux08}, DNA detection during translocation through nuclear pores \cite{Auger14},\ldots 

The present work proposes a generic nonlinear signal/image filtering unified framework for the analysis of piecewise homogeneous (piecewise-constant, piecewise-linear) experimental datasets.
The major challenges here are both to obtain fast algorithms so as to handle the large amount of data that need to be analyzed to yield accurate and relevant conclusions (e.g., in producing a phase diagram or in analyzing  video frames of large size images) and to be able to perform an automated and data-driven tuning of hyperparameters, unavoidably entering any nonlinear filtering procedure, and whose arbitrary selection (by expert visual inspection) might have drastic impacts on achieved outcomes and thus on a posteriori drawn physical interpretations. 

The unified signal/image nonlinear filtering framework proposed here is based on proximal schemes~\cite{combettes_signal_2005,bauschke_convex_2011} to obtain fast algorithms, and on the Stein unbiased estimator framework to design an automated data-driven hyperparameter tuning procedure. 

Section~\ref{s:ipgen} is dedicated to the formulation of this framework as an inverse problem, and recalls state-of-the-art strategies with focus on piecewise constant or linear estimation both in signal or images. 
Section~\ref{s:propmet} details the proposed algorithmic framework. 
Section~\ref{s:res} illustrates the performance on the two examples discussed above, solid friction and multiphase flows.

A documented toolbox (in {\sc Matlab}), for the implementation of this signal/image processing nonlinear filtering and data-driven hyperparameter tuning, is freely available at \\
\texttt{https://github.com/bpascal-fr/stein-piecewise-filtering}.

\section{Non linear filtering formulated as inverse problems} 
\label{s:ipgen}

\subsection{Direct models}

Let $S = \{\underline{n} = (n_1,n_2): 1\leq n_1 \leq N_1, 1\leq n_2 \leq N_2\}$ denote a lattice, supporting $\overline{x} = (\overline{x}_{\underline{n} })_{\underline{n} \in S}$, the unknown signal/image of size $ N = N_1\times N_2$ ($N_1=1$ for univariate 1D signal analysis and $N_1=K$ for  multivariate 1D signal analysis with $K$ components). 
Observation $z = (z_{m_1,m_2})_{1\leq m_1\leq M_1, 1\leq m_2\leq M_2}$ is of size $M = M_1\times M_2$ consists of a degraded version of $\overline{x}$, which stems  both from a linear degradation (e.g. filtering), denoted $A\in \mathbb{R}^{M\times N}$, and additive random noise, denoted $\mathcal{B}$. 

Handling an inverse problem relies first on an accurate design of the observation (or direct) model, which here takes the following form:
$$
z = \mathcal{B}(A \overline{x}).
$$

In this work, $S$ corresponds to an homogeneous neighborhood system. 
For instance, when considering a 1D signal, a site $n_2\in\{2,\ldots,N_1-1\}$ has two nearest neighbors $\mathcal{N}_{n_2} = \{n_2-1, n_2+1\}$. In a general regular rectangular lattice $S$ and for a 4-neighborhood system, every interior point has four neighbors that yields to $\mathcal{N}_{\underline{n} } = \{(n_1-1,n_2),(n_1+1,n_2),(n_1,n_2-1),(n_1,n_2+1)\}$. The pair $(S,\mathbb{E})$ constitutes a graph where $S$ contains the nodes and $\mathbb{E}$ determines the links between the nodes according to the neighboring relationship. \\

We detail this direct model on the two nonlinear physics problems described in Introduction (low confinement solid friction and porous media multiphase flows) and illustrated in Figure~\ref{fig:phyexp}. 
For solid friction, the challenging question consists in denoising obserbation $z$ (Figure~\ref{fig:phyexp}(a)), with $A = \text{Id}$ and the presence of additive impulsive noise.  
For multiphase flows, the challenging issues pertains to segmentation of textures such as the one in Figure~\ref{fig:phyexp}(b). 
In such a case, information $\overline{x}$ refers to piecewise constant scale-free texture features, and observation $z$ is obtained from a nonlinear multiscale transform (cf. Section~\ref{ss:textseg}).
Noise is assumed additive and Gaussian, with spatial and multiscale correlations.

\subsection{State-of-the-art}

Solving an inverse problem consists in providing an estimator $\widehat{x}$, as close as can be from information $\overline{x}$. 
This has been largely addressed in the literature and we propose first a brief overview of the main inverse problem solving streams (see also \cite{pustelnik2012_j-ieee-tsp_sur_ads,Cai_JF_2012_j-ams_ima_rtv}), before focusing, second, on the specific assumptions on the model required to design parameter-free and fast nonlinear piecewise filtering.\\

\noindent \textbf{Bayesian arguments and most standard models} --  On the first hand, Markov Random Fields (MRF) have been introduced in visual labelling to establish probabilistic distributions of interacting labels, aiming to analyze dependencies of a physical phenomena \cite{Li_S_2009_book_MRF}. In such a formalism $\overline{x}$ and $z$ are considered as realizations of random vectors $\overline{X}$ and $Z$ defined on the set $S$. $X$ is said to be a MRF on $S$ with respect to a neighborhood system $\mathbb{E}$ if and only if positivity (i.e. $P(X=x)>0$) and Markovianity $P(x_{\underline{n}}  \vert x_{\{S\}-{\underline{n}}}) = P(x_{\underline{n}} \vert x_{\mathcal{N}_{\underline{n}}})$, which models the local characteristics of $X$, are satisfied. Other properties such as homogeneity and isotropy can be depicted. 

The link between the MRF, characterized by its local properties, and another standard random field, the Gibbs random field, characterized by global properties, has been provided by Hammersley and Clifford \cite[Theorem~1]{Moller_J_book_spatialstat}.  We recall that a Gibbs distribution relative to the graph $\{S, \mathbb{E}\}$ is a probability measure and it has the following representation :
\begin{equation}
P(\omega) = \frac{1}{C} e^{-U(\omega)/T} 
\end{equation}
where $C$ is the normalizing constant  called the partition function such that  $C = \sum_{\omega} e^{-U(\omega)/T}
$ and $T$ stands for temperature, which controls the sharpness of the distribution. High temperature leads to all configurations equally distributed, while a temperature close to 0 concentrates the distribution around the global energy minima. $U(\omega)$ denotes the energy function 
and  $P(\omega)$ measures the probability of the occurence of a specific configuration $\omega$. The more probable configurations are those with the lower energies.   
The terminology comes from statistical physics where such measures are equilibrium states for physical systems (e.g. ferromagnets). $U(\omega)$ can be formulated with contributions from external fields (i.e. $x_{n_1,n_2}$) and pair interactions (e.g. $x_{n_1,n_2}x_{n_1+1,n_2}$). For instance, the Ising model reads
\begin{equation}
U(x) = -\alpha \sum x_{n_1,n_2} - \beta \Bigg(\sum x_{n_1,n_2}x_{n_1+1,n_2} + \sum x_{n_1,n_2}x_{n_1,n_2+1}\Bigg)
\end{equation}
considering $\omega=x$ and for some parameters $\alpha\in \mathbb{R}$ and $\beta>0$, which measure, the external magnetic moment and bonding strengths.

Geman and Geman \cite{geman1987:bayesian_resto} handle the maximization of the conditional probability distribution of $(x,e) \in \{S,\mathbb{E}\}$ given the data $z$ (i.e. find the mode of the posterior distribution), which is known as the maximum a posteriori estimation or penalized maximum likelihood. In \cite{geman1986bayesian,geman1987:bayesian_resto,Geman_D_1995}, the authors prove that the posterior is a Gibbs distribution over $\{S,\mathbb{E}\}$ with energy function 
\begin{equation}
\label{eq:gemangeman}
U(x,e) = \frac{1}{2\sigma^2} \Vert A x  - z\Vert^2 + \beta \sum_{\underline{n},\underline{n}'\in \mathcal{N}_{\underline{n}}}  \varphi(x_{\underline{n}} - x_{\underline{n}'})(1-e_{\underline{n},\underline{n}'}) + \alpha \psi(e)
\end{equation}
so that $\omega = (x,e)$ when $\mathcal{B}$ designates an additive zero-mean Gaussian noise with a variance $\sigma^2$.   $e\in \mathbb{E}$ denotes the coded line states and $\varphi(\eta) = -1$ if $\eta=0$ and 1 if $\eta\neq 0$. The first term acts as a data fidelity term and forces the approximation $x$ to be close to $z$, the second term allows small variations of $x$ except at the locations where $e_{\underline{n},\underline{n}'}=1$, and $\psi$ is constructed to organize the line process. Finally, $\alpha > 0$ and $\beta > 0$ denote regularization parameters controlling the smoothness of the solution and the length of the interfaces. This model can be interpreted as a coupled MRF dealing jointly with image restoration and edges detection: one MRF for the pixel values and one for the edges values that are described respectively in the image lattice or in its dual lattice.  This model has strong link with the continuous Mumford-Shah setting~\cite{Chambolle_A_1995}.

For specific choices of $\varphi$ and $\psi$ \cite{Geman_D_1995,Lobel_P_1997}, an alternative equivalent formulation is the Blake-Zisserman model formulated as
\begin{equation}
U(x) =\frac{1}{2\sigma^2} \Vert A x  - z\Vert^2 + \lambda \sum_{\underline{n},\underline{n}'\in \mathcal{N}_{\underline{n}}} \min \Big((x_{\underline{n}} - x_{\underline{n}'})^2, \eta\Big) 
\end{equation}
where $\lambda, \eta>0$, leading to the so-called truncated $\ell_2$  and whose interest is to favor piecewise smooth solution. Another model very close is the Potts model that can be interpreted as a $\ell_0$-penalization over $x_{\underline{n}} - x_{\underline{n}'}$ that is designed to provide piecewise-constant estimates \cite{storath2015:Potts}. For numerical reasons detailed below,  the most standard convex relation  is the anisotropic total-variation penalization  which reads~\cite{rudin1992nonlinear,chambolle2011first}:
\begin{equation}
\label{eq:tvenergy}
U(x) = \frac{1}{2\sigma^2} \Vert A x  - z\Vert^2 + \lambda \sum_{\underline{n},\underline{n}'\in \mathcal{N}_{\underline{n}}}  \vert x_{\underline{n}} - x_{\underline{n}'} \vert.
\end{equation}

\noindent \textbf{Solving inverse problems} -- Once an energy (or functional) has been designed to fit the considered problem, numerical strategies have to be designed to implement both the hyperparameter selection and the computation of the minimizing solution $x$, also corresponding to the most probable $\omega$ or moments of $P$. 

On one hand, Markov Chain Monte Carlo algorithms address simulations from a probability distribution $P$. The function $P$ can be written in a closed-form expression but the objective is generally to access the moments of $P$, which are not computable analytically. The two main techniques used in MCMC are Metropolis-Hastings, which relies on accept/reject mechanism and Gibbs sampler, which simplifies the high dimensional problem by successively simulating from different smaller dimensional components. The main limitation of these techniques is to be computationally intensive for solving large size inverse problems (see a contrario~\cite{Marnissi_Y_2018_j-entropy_aux_vm,Vacar_C_2019_j-jasp_uns_jds}).  We should also refer to some specific configurations where a closed form expression is available such as for the Ising model in 1D and 2D but which is not adapted for general inverse problem solving considered in this work. 

When one wants to estimate jointly the maximum of a posteriori and its hyperparameters, Bayesian hierarchical inference frameworks are particularly adapted and received considerable interest for addressing change-point detection or piecewise denoising problems \cite{Dobigeon_N_2007_j-tsp_joi_spc,Dobigeon_N_2007_j-csda_joi_sws,Pereyra_M_2013_j-tip_est_gcp} or texture segmentation \cite{Vacar_C_2019_j-jasp_uns_jds}. However, to the best of our knowledge, for the proposed unified 1D-2D framework considered in this work, such a general efficient strategy has not yet been designed.

 On other hand, during the last twenty years, important research efforts  have been dedicated to convex but non-smooth energy generally formulated as a sum of two or three terms: a data-fidelity term, a penalization and a constraint \cite{Combettes2011,bauschke_convex_2011,condat_primal-dual_2013}.  This framework is thus especially adapted when dealing with an energy such as \eqref{eq:tvenergy}. These algorithmic strategies are particularly efficient when dealing with $A$ full-rank which is rarely the case in standard restoration/reconstruction problems but which is more encountered in experimental physics processing when dealing either with denoising i.e. $A= \mathrm{Id}$ such as in friction experiments or when dealing with texture reconstruction especially adapted to study multiphase flow dynamics as it will be described later. However, when one handles such optimization strategy to find the minimizer of the energy $U$, the selection of the automated regularization parameter(s) is not addressed. 
For automated selection, one could consider either an empirical rule that consists in setting $\lambda \sim N^{1/2} \sigma/4$, with $N$ the signal size and $\sigma$ the noise standard deviation, estimated e.g., from the median value of the absolute value of the wavelet coefficients \cite{donoho1994ideal}, or an hybrid Bayesian hierarchical inference framework combined with $\ell_0$-minimization startegy \cite{Frecon_J_2017_j-tsp_bay_sl2} in the specific case of piecewise-constant denoising, or the recourse to Stein Unbiased Risk Estimator (SURE) which provides an unbiased estimator of the mean square error \cite{Benazza05,ramani2008monte,deledalle2014stein}. 
Our contribution focuses on such a strategy, its implementability, and its applicability on real physics experiments.

\section{Proximal operator based nonlinear filtering: fast algorithms and automated data-driven hyperparameter tuning}
\label{s:propmet}

\subsection{Nonlinear filtering formulation}  

In this work, we consider an estimator $\widehat{x}(z; \Lambda)$ of the quantity of interest $\bar{x}$, from a corrupted observation $z$, parametrized by $\Lambda$. 
The estimate is obtained from the minimization of an energy, inspired from \eqref{eq:tvenergy}, and defined as:
\begin{equation}
\label{eq:nl}
\widehat{x}(z; \Lambda) \in \underset{x\in\mathbb{R}^{N}}{\textrm{Argmin}}\; \Vert A x -z\Vert_2^2 +   \Vert D_\Lambda x \Vert,
\end{equation}
where the matrix $D_{\Lambda}$ models a weighted discrete differentiation operator, so that the penalization $\lVert D_{\Lambda} x \rVert$ enforces some regularity of the estimate $\widehat{x}(z; \Lambda)$. Specific instances of \eqref{eq:nl} are provided :
\begin{itemize}
\item To favor joint piecewise \textit{constancy} of K multivariate signals, the regularization parameters are stored in a vector $\Lambda = \left( \lambda_1, \hdots, \lambda_K\right) \in \left(\mathbb{R}_+^*\right)^K$, and the operator $D_{\Lambda}$ is a first order differentiation operator, acting componentwise, also called discrete gradient, writing for $k\in \lbrace 1, \hdots, K\rbrace$, $n_2 \in \lbrace 1, \hdots, N_2 - 1\rbrace$,
\begin{align}
\left( D_{\Lambda} x \right)_{k,n_2} = \lambda_k \left( x_{k,n_2+1} - x_{k,n_2} \right) 
\end{align} 
and where $\lVert \cdot \rVert$ is the mixed $\ell_{1,2}$-norm
\begin{align}
\lVert D_{\Lambda} x \rVert_{1,2} = \sum_{n_2 = 1}^{N_2-1} \sqrt{\sum_{k = 1}^K \left( D_{\Lambda} x \right)_{k,n_2}^2  }.
\label{eq:l12norm}
\end{align}
\item Enforcing joint piecewise \textit{linearity} requires a second order differentiation operator, named discrete Laplacian, defined for $k\in \lbrace 1, \hdots, K\rbrace$ and $n_1 \in \lbrace 2, \hdots, N_1 - 1\rbrace$,
\begin{align}
\label{eq:l12}
\left( D_{\Lambda} x \right)_{n_1,k} = \lambda_k \left( x_{n_1+1,k} - 2x_{n_1,k} +x_{n_1-1,k} \right)
\end{align}
 composed with the $\ell_{1,2}$-norm defined in~\eqref{eq:l12norm}.\\
\item Image segmentation is performed imposing piecewise constancy prior, using a two dimensional difference operator.
For an image $x \in \mathbb{R}^{N_1\times N_2}$ the horizontal and vertical gradients are computed for each pixel with $1\leq n_1\leq N_1-1$ and $1 \leq n_2 \leq N_2-1$
\begin{align}
\label{eq:diff_2D}
\left( D_{\Lambda} x \right)_{n_1,n_2} = \lambda \begin{pmatrix}
\left(D_1x\right)_{n_1,n_2}  \\
\left(D_2x\right)_{n_1,n_2} 
\end{pmatrix}= \lambda \begin{pmatrix}
 x_{n_1,n_2+1}-x_{n_1,n_2}  \\
x_{n_1+1,n_2}-x_{n_1,n_2} 
\end{pmatrix}
\end{align}
and coupled into an $\ell_{1,2}$-norm 
\begin{align}
\label{eq:l12_2D}
\lVert D_{\Lambda} x \rVert = \lambda \sum_{n_1 = 1}^{N_1-1} \sum_{n_2 = 1}^{N_2-1} \sqrt{\left(D_1x\right)_{n_1,n_2}^2 +\left(D_2x\right)_{n_1,n_2} ^2 } := \lambda \mathrm{TV}(x).
\end{align}
The above penalization is known as the isotopic Total Variation \cite{rudin1992nonlinear}. \\
\end{itemize}

The estimate $\widehat{x}(z;\Lambda)$ is the result of a trade-off between the fidelity to the observation model and some regularity priors and the balance is tuned by the hyperparameter $\Lambda$.
Hence, our purpose is twofold.
First, solving the minimization Problem~\eqref{eq:nl}, that is, for fixed hyperparameter $\Lambda$, given an observation $z$, compute $\widehat{x}(z; \Lambda)$ the minimizer of~\eqref{eq:nl}.
Second, finding the \textit{best} hyperparameter $\Lambda^{\dagger}$ minimizing the quadratic error $ \mathbb{E}\{\Vert \widehat{x}(z; \Lambda) - \overline{x} \Vert^2 \}$, i.e. 
\begin{problem}
\label{pb:min_R}
Find 
\begin{align}
\Lambda^{\dagger}= \underset{\Lambda \in \left( \mathbb{R}_+^*\right)^K}{\arg\min}  \, \, \mathbb{E}\{\Vert \widehat{x}(z; \Lambda) - \overline{x} \Vert^2 \} 
\end{align}
 where $\widehat{x}(z; \Lambda)$  is defined by \eqref{eq:nl} and the expectation is taken over all realizations of the noise $\mathcal{B}$ corrupting the observation $z = \mathcal{B}(A\bar{x})$.
\end{problem}

In the next sections, we specify the assumptions over $A$ and $\mathcal{B}$ allowing us to derive a fast algorithmic scheme to estimate $\Lambda^{\dagger}$.

\subsection{Convex non-smooth minimization}
\label{subsec:convex_min}

The objective function appearing in Problem~\eqref{eq:nl} is convex, since the composition of a linear operator and a norm is convex.
Yet, because of the presence of the norm $\lVert \cdot \rVert$, it is non-smooth. 
Consequently, the minimization of~\eqref{eq:nl} requires proximal algorithms~\cite{bauschke_convex_2011,Combettes2011,parikh2014proximal}, which in general suffer from low  convergence rate.
Nevertheless, provided that the operator $A$ is injective, it is possible to design accelerated primal-dual schemes~\cite{chambolle2011first} and to obtain linear convergence rate toward the minimizer of~\eqref{eq:nl}.
Such algorithms relies on proximity operators~\cite{parikh2014proximal}, whose definition is recalled in Definition~\ref{def:prox}.
Further, disposing closed-form expressions of the proximity operators of the data fidelity term and the penalization function  is a key element to design fast implementations of primal-dual algorithms.

\begin{definition}
\label{def:prox}
For a convex lower semi-continuous function $\varphi : \mathbb{R}^N \rightarrow \mathbb{R}\cup\{+\infty\}$ and $\tau>0$, the proximity operator is defined as
\begin{align}
(\forall x \in \mathbb{R}^N) \qquad \mathrm{prox}_{\tau \varphi}(x) = \underset{\widetilde{x}}{\arg\min} \, \frac{1}{2} \lVert \widetilde{x} - x \rVert^2 + \tau \varphi(\widetilde{x})
\end{align}
where $\lVert \cdot \rVert$ denotes the Euclidean norm on $\mathbb{R}^N$.
\end{definition}

Few examples of well-established closed-form expression for proximity operator of interest in this work are recalled.
\begin{example}
\label{ex:proxs}
 The proximity operator of the data fidelity term $\Vert A \cdot - z \Vert_2^2$ as a closed form expression that is, for every $\tau>0$,
\begin{align}
(\forall x \in \mathbb{R}^N) \qquad \mathrm{prox}_{\tau \lVert A\cdot -z\rVert^2}(x) = \left( \mathrm{Id} + 2 \tau A^{\top} A \right)^{-1} \left( x + 2\tau A^{\top}z\right).
\end{align}
\end{example}
\begin{example}
The proximity operator of the multivariate 1D $\ell_{1,2}$-norm defined in~\eqref{eq:l12norm} is, for every $y \in \mathbb{R}^{K\times N_2}$,
\begin{align}
\left( \mathrm{prox}_{\tau \lVert \cdot \rVert{1,2}}(y) \right)_{k,n_2} = \left\lbrace 
\begin{array}{ll}
\left( 1 - \frac{\tau}{\lVert y_{\cdot, n_2}\rVert_2}\right) y_{k,n_2} & \text{if} \, \lVert y_{\cdot, n_2}\rVert_2 > \tau, \\
0 & \text{otherwise},
\end{array}
\right. 
\end{align}
where $\lVert y_{\cdot,n_2} \rVert_2 := \sqrt{\sum_{k = 1}^K y_{k,n_2}^2  }$.
\end{example}
\begin{example} The proximity operator of the 2D $\ell_{1,2}$-norm of Equation~\eqref{eq:l12_2D}, for $y = \left( y^{(H)}, y^{(V)}\right) \in \mathbb{R}^{2\times N_1\times N_2}$, 
\begin{align}
\left( \mathrm{prox}_{\tau \lVert \cdot \rVert{1,2}}(y) \right)_{n_1,n_2} = \left\lbrace 
\begin{array}{ll}
\left( 1 - \frac{\tau}{\lVert y_{n_1, n_2}\rVert_2}\right) y_{n_1,n_2} & \text{if} \, \lVert y_{n_1, n_2}\rVert_2 > \tau, \\
0 & \text{otherwise},
\end{array}
\right. 
\end{align}
where $\lVert y_{n_1,n_2} \rVert_2 := \sqrt{(y^{(H)}_{n_1,n_2})^2 +(y^{(V)}_{n_1,n_2})^2  }$.
\end{example}

\begin{theorem}
\label{thm:CP}
Assuming that the deformation operator $A$ is injective and denoting by $\mu > 0$ the smallest eigenvalue of $2A^{\top} A$, the sequence $\left( x^{[t]} \right)_{t \in \mathbb{N}}$ defined in Algorithm~\ref{algo:pd_sure} converges toward the solution $\widehat{x}(z; \Lambda)$ of
\begin{align}
\label{eq:min_pd}
 \underset{x\in\mathbb{R}^{N}}{\textrm{minimize}}\; \Vert A x -z\Vert_2^2 +   \Vert D_\Lambda x \Vert.
\end{align}
Further, it has been proven in~\cite{chambolle2011first} that, for any $\epsilon > 0$, there exists $t_0$ such that 
\begin{align}
(\forall t \geq t_0) \quad \left\lVert \widehat{x}(z; \Lambda) - x^{[t]} \right\rVert \leq \frac{1+\epsilon}{t^2} \left( \frac{\lVert \widehat{x}(z; \Lambda) - x^{[0]}\rVert^2}{\mu^2\tau_0^2} + \frac{\lVert A\rVert^2\lVert \widehat{y}(z; \Lambda) - y^{[0]}\rVert^2}{\mu^2} \right)
\end{align}
where $\widehat{y}(z; \Lambda)$ denotes the solution of the dual problem of Problem~\eqref{eq:min_pd}.
Hence, the convergence rate of the iterates $\left(x^{[t]}\right)_{t \in \mathrm{N}}$ scales like $\mathcal{O}(1/t^2)$.
\end{theorem}

\begin{proof}
This theorem is a direct	application of~\cite[Theorem~2]{chambolle2011first}, stated and demonstrated for the minimization of objective functionals defined as the sum of a $\mu$-strongly convex data fidelity and convex, proper, lower semi-continuous penalization, which is the case here.\\
Indeed, thanks to the assumption that $A$ is full-rank, the considered data fidelity term $\lVert Ax-z\rVert_2^2$ is $\mu$-strongly convex with modulus $\mu = 2\min \mathrm{Sp} (A^{\top} A) > 0$.
Further, the penalization being the composition of a linear operator and a norm is satisfies the aforementioned conditions.\\
Then, the primal-dual updates of Algorithm~\ref{algo:pd_sure} corresponds to the customization of the Algorithm~2 of~\cite{chambolle2011first} to the problem of finding $\widehat{x}(z,\Lambda)$ solution of~\eqref{eq:nl}, hence corresponding to $G(x) = \lVert Ax - z \rVert^2$, linear operator $K = D_{\Lambda}$ and $F = \lVert \cdot \rVert_{1,2}$.
\end{proof}

We have to note that, because of the operation $\left( \mathrm{Id} + 2 \tau A^{\top} A \right)^{-1}$, the proximity operator of the data-term might be uneasy to evaluate. However, for numerous configuration, this expression as a closed form expression. First when $A = \textrm{Id}$. Second, when $A$ is diagonalizable in a specific basis such as Fourier for circulant matrices (leading to $\mathcal{O}(MN)$ operations). Another specific choice of $A$ will be discussed in Section~\ref{ss:textseg}.

\subsection{Stein Unbiased Risk Estimate}

Once an efficient algorithmic strategy has been identified to estimate $\widehat{x}(z; \Lambda)$, the second major difficulty raised by Problem~\ref{pb:min_R} is that, in practice, one does not have access to the true signal/image $\bar{x}$ and hence cannot compute $\mathbb{E}\left\lbrace\Vert \widehat{x}(z; \Lambda) - \overline{x} \Vert^2 \right\rbrace$.
To handle this limitation, Stein proposed an Unbiased Risk Estimator~\cite{stein1981estimation}, denoted $\mathrm{SURE}(\Lambda)$, providing an usable approximation of the quadratic risk in the case of i.i.d. Gaussian noise.
This estimator was then widely extended to more general noise models~\cite{eldar2008generalized,pascal2020automated}.
\begin{theorem}[Stein Unbiased Risk Estimate]
\label{th:sure}
We denote $\widehat{x}(z; \Lambda)$ the parametric estimator defined in \eqref{eq:min_pd} of the ground truth $\bar{x}$ from observation $z = \mathcal{B}(A\bar{x})$ corrupted by a full-rank deformation operator $A$ and additive (possibly correlated) Gaussian noise $\mathcal{B}$, with covariance matrix $\mathcal{S}$.
The Stein Unbiased Risk Estimate, defined as 
\begin{align}
\label{eq:def_sure}
 \mathrm{SURE}(\Lambda) := \left\lVert \Phi \left(A\widehat{x}(z; \Lambda) - z \right)\right\rVert^2 + 2\mathrm{Tr}\left( \mathcal{S} \Phi^{\top}  \frac{\partial \widehat{x}(z; \Lambda)}{\partial z}\right) - \mathrm{Tr}(\Phi  \mathcal{S} \Phi^{\top} ),
\end{align}
satisfies the following unbiasedness property
\begin{align}
\mathbb{E}\left\lbrace \mathrm{SURE}(\Lambda) \right\rbrace = \mathbb{E}\{\Vert \widehat{x}(z; \Lambda) - \overline{x} \Vert^2 \}.
\end{align}
where $\Phi := \left(A^{\top} A\right)^{-1}A^{\top}$ and $\partial \widehat{x}(z; \Lambda)/\partial z$ denotes the Jacobian of $\widehat{x}(z; \Lambda)$ w.r.t. observation $z$.
\end{theorem}
\begin{proof}
A complete and detailed proof was proposed in~\cite{pascal2020automated}.
\end{proof}

\begin{definition}[Degrees of freedom]
The second term in the definition of $\mathrm{SURE}(\Lambda)$, in Equation~\eqref{eq:def_sure},
\begin{align}
2\mathrm{Tr}\left( \mathcal{S} \Phi^{\top}  \frac{\partial \widehat{x}(z; \Lambda)}{\partial z}\right)
\end{align}
is called the \textit{degrees of freedom}.
\end{definition}

Since $\Phi \in \mathbb{R}^{M\times N}$ and $\partial \widehat{x}(z; \Lambda)/\partial z \in \mathbb{R}^{N \times M}$, are large size matrices computing the trace of $\Phi^{\top}  \partial \widehat{x}(z; \Lambda)/\partial z$ is very expansive and hence the evaluation of the degrees of freedom requires additional tools.
This difficulty is overcome using, a Monte Carlo (MC) strategy, only requiring the evaluation of the Jacobian on a random vector $\delta \in \mathbb{R}^M$.
Hence, it is only necessary to store a vector of size $N$, instead of manipulating a matrix of size $N \times M$, which decreases drastically the computational and memory costs.
Moreover, when $\widehat{x}(z; \Lambda)$ is obtained from a minimization scheme, there is often no closed-form expression of the Jacobian, hence we will approximate $\partial \widehat{x}(z; \Lambda)/\partial z [\delta]$ using Finite Difference (FD) approximation of the Jacobian.
Altogether, Monte Carlo and Finite Difference strategies lead to the following FDMC Stein Unbiased Risk Estimate.
\begin{theorem}[Finite Difference Monte Carlo SURE]
\label{thm:sure}
Let $\widehat{x}(z; \Lambda)$ denote a parametric estimator of ground truth $\bar{x}$ from observation $z = \mathcal{B}(A\bar{x})$ corrupted by a full-rank deformation operator $A$ and additive (possibly correlated) Gaussian noise $\mathcal{B}$, with covariance matrix $\mathcal{S}$ and $ \Phi = \left(A^{\top} A\right)^{-1}A^{\top}$.
The FDMC Stein Unbiased Risk Estimate is defined as
\begin{align}
\mathrm{SURE}_{\varepsilon, \delta}(\Lambda) &:=  \left\lVert \Phi \left(A\widehat{x}(z; \Lambda) - z \right)\right\rVert^2 \nonumber \\
&+  \frac{2}{\varepsilon}\left\langle \Phi^{\top} \left( \widehat{x}(z+\varepsilon \delta; \Lambda) - \widehat{x}(z; \Lambda) \right), \mathcal{S}\delta \right\rangle 
\label{eq:SURE}- \mathrm{Tr}(\Phi  \mathcal{S} \Phi^{\top} ).
\end{align}
Provided that $\widehat{x}(z; \Lambda)$ is uniformly Lipschitz w.r.t. observation $z$ and integrable against Gaussian density, $\mathrm{SURE}_{\varepsilon, \delta}(\Lambda)$ is an asymptotically unbiased estimator of the quadratic risk, i.e.
\begin{align}
\lim\limits_{\varepsilon \rightarrow0} \mathbb{E}\left\lbrace \mathrm{SURE}_{\varepsilon, \delta}(\Lambda) \right\rbrace = \mathbb{E}\{\Vert \widehat{x}(z; \Lambda) - \overline{x} \Vert^2 \},
\end{align}
where the expectation in the left hand side is taken over both the noise $\mathcal{B}$ and the Monte Carlo vector $\delta \sim\mathcal{N}(0, \mathrm{Id})$.
\end{theorem}
\begin{proof}
The proof directly follows from Theorem~2 in~\cite{pascal2020automated}.
\end{proof}

Thus, Problem~\ref{pb:min_R} is replaced by 
\begin{problem}
\label{pb:min_SURE}
Find $\widehat{\Lambda}^{\dagger}= {\arg\min}_{\Lambda \in \left( \mathbb{R}_+^*\right)^K} \, \mathrm{SURE}_{\varepsilon, \delta}(\Lambda) $.
\end{problem}

\subsection{Automated data-driven hyperparameter tuning}
\label{sec:selec_lambda}

In order to solve Problem~\ref{pb:min_SURE}, two strategies can be considered. First a grid search, computing $\mathrm{SURE}_{\varepsilon, \delta}(\Lambda)$ over a large range of hyperparameter values, corresponding to the discrete set $\boldsymbol{\Lambda} =\left(\Lambda_i\right)_{i = 1}^I$ and selecting \textit{a posteriori} the hyperparameters of the grid for which $\mathrm{SURE}_{\varepsilon, \delta}(\Lambda_i)$ is minimal, denoted $\widehat{\Lambda}_{\mathrm{G}}$, as described in Algorithm~\ref{alg:grid}.
The major drawback of this approach is its computational cost, increasing algebraically with the dimension of the hyperparameters $\Lambda$.\\

\begin{algorithm}[h!]
\caption{Accelerated primal dual scheme for minimization of~\eqref{eq:nl}.\label{algo:pd_sure}}
\begin{algorithmic}
\REQUIRE Set $\varepsilon>0$, $\delta\in \mathbb{R}^M$, $\tau_0 > 0$, $\sigma_0 > 0$, such that $\tau_0\sigma_0 \lVert D_{\Lambda} \rVert^2 < 1$. .
\FOR{$\widetilde{z} = \left\lbrace z, z+\varepsilon\delta \right\rbrace$}
\STATE Choose $\widetilde{x}^{[0]}\in \mathbb{R}^N$, $x^{[0]}\in \mathbb{R}^N$, $y^{[0]} = D_{\Lambda} x^{[0]}$
\STATE $\partial_{\Lambda} \widetilde{x}^{[0]} \leftarrow 0_N$
\STATE $\partial_{\Lambda} x^{[0]} \leftarrow 0_N$
\STATE $\partial_{\Lambda} y^{[0]} \leftarrow D_{\Lambda} \partial_{\Lambda} x^{[0]}$
\FOR{$t = 0$ \TO $T_{\max}-1$}
\STATE \COMMENT{Primal-dual updates}
\STATE $y^{[t + 1]} = \mathrm{prox}_{\sigma_t \left(\lVert \cdot \rVert \right)^*} \left( y^{[t]} + \sigma_t D_{\Lambda} \widetilde{x}^{[t]}\right)$
\STATE $x^{[t + 1]} = \mathrm{prox}_{\tau_t \lVert A \cdot - \widetilde{z} \rVert^2} \left( x^{[t]} - \tau_t D_{\Lambda} y^{[t+1]}\right)$
\STATE $\vartheta_t = \sqrt{1 + 2\mu \tau_t}$, $\tau_{t+1} = \tau_t/\vartheta_t$ and $\sigma_{t+1} = \vartheta_t \sigma_t$
\STATE $\widetilde{x}^{[t+1]} = x^{[t+1]} + \vartheta_t\left( x^{[t+1]} - x^{[t]}\right)$
\ENDFOR
\STATE  $\widehat{x}(\widetilde{z} ; \Lambda) \leftarrow x^{[T_{\max}]}$ 
\ENDFOR
\STATE Compute $\mathrm{SURE}_{\varepsilon, \delta}(\Lambda)$ injecting $\widehat{x}(z ; \Lambda)$ and $\widehat{x}(z+\varepsilon\delta ; \Lambda)$ in Formula~\eqref{eq:SURE}
\RETURN{$\mathrm{SURE}_{\varepsilon, \delta}(\Lambda)$}
\end{algorithmic}
\end{algorithm}

 \begin{algorithm}[h!]
\caption{\label{alg:grid} Grid search for SURE minimization.}
\begin{algorithmic}
\REQUIRE Grid $\boldsymbol{\Lambda} = \left(\Lambda_i\right)_{i=1}^I$, $\varepsilon > 0$, $\delta \sim\mathcal{N}(0, \mathrm{Id}) \in \mathbb{R}^M$, 
\FOR{$i = 1$ \TO $I$}
\STATE $\mathrm{ERROR}(i) \leftarrow \mathrm{SURE}_{\varepsilon,\delta}(\Lambda_i)$, computed from Algorithm~\ref{algo:pd_sure}
\ENDFOR
\STATE $\widehat{i}_{\mathrm{G}} \leftarrow \underset{1 \leq i \leq I}{\arg\max} \, \, \mathrm{ERROR}(i)$
\RETURN $\widehat{\Lambda}_{\mathrm{G}} = \Lambda_{\widehat{i}_{\mathrm{G}}}$
\end{algorithmic}
\end{algorithm}

In order to provide faster implementations, we consider automated selection of hyperparameters.
To that aim, the number of hyperparameters is assumed to be $K = \mathcal{O}(1)$ and hence quasi-Newton algorithms are appropriate since they can handle very efficiently minimization in low dimension.
It requires to compute the gradient of $\mathrm{SURE}_{\varepsilon, \delta}(\Lambda)$ w.r.t. the hyperparameter $\Lambda$.
For this purpose, it was proposed a Stein Unbiased GrAdient Risk estimate, denoted $\mathrm{SUGAR}_{\varepsilon, \delta}(\Lambda)$~\cite{deledalle2014stein,pascal2020automated}, which, under the conditions of Theorem~\ref{thm:sure}, writes
\begin{align}
\mathrm{SUGAR}_{\varepsilon, \delta}(\Lambda) &:= 2 \left( \Phi A \frac{\partial \widehat{x}(z; \Lambda)}{\partial \Lambda} \right)^{\top} \left( \Phi \left(A\widehat{x}(z; \Lambda) - z \right)\right)\nonumber \\
\label{eq:SUGAR}&+ \frac{2}{\varepsilon}\left\langle \Phi^{\top} \left( \frac{\partial \widehat{x}(z+\varepsilon \delta; \Lambda)}{\partial \Lambda} - \frac{\partial \widehat{x}(z; \Lambda)}{\partial \Lambda} \right), \mathcal{S}\delta \right\rangle.
\end{align}

A sketch of quasi-Newton descent~\cite{nocedal2006numerical}, particularized to Problem~\ref{pb:min_SURE}, is detailed in Algorithm~\ref{alg:BFGS}.
It generates a sequence $\left(\Lambda^{[j]}\right)_{j \in \mathbb{N}}$ converging toward a minimizer of $\mathrm{SURE}_{\varepsilon, \boldsymbol{\delta}}(\Lambda )$, denoted $\widehat{\Lambda}_{\mathrm{BFGS}}$.\\
This algorithm relies on a gradient descent step involving a descent direction $d^{[j]}$ obtained from the product of BFGS approximated inverse Hessian matrix $H^{[j]}$ and the gradient $\mathrm{SUGAR}_{\varepsilon, \delta}(\Lambda)$ obtained from Algorithm~\ref{algo:pd_sugar}.
The descent step size $\alpha^{[j]}$ is obtained from a line search which stops when Wolfe conditions are fulfilled~\cite{nocedal2006numerical,curtis2017bfgs}.
Finally, the approximated inverse Hessian matrix $H^{[j]}$ is updated according to a BFGS strategy.
\begin{remark}
The line search is the most time consuming.
Indeed, the routines $\mathrm{SURE}$ and $\mathrm{SUGAR}$ are called for several hyperparameters of the form $\Lambda^{[j]} + \alpha d^{[j]}$, each call requiring to run differentiated primal-dual scheme twice.
\end{remark}

\label{subsec:auto_lambda}

\begin{algorithm}[h!]
\caption{Accelerated primal dual scheme for minimization of~\eqref{eq:nl} with iterative forward differentiation.\label{algo:pd_sugar}}
\begin{algorithmic}
\REQUIRE Set $\varepsilon>0$, $\delta\in \mathbb{R}^M$, $\tau_0 > 0$, $\sigma_0 > 0$, such that $\tau_0\sigma_0 \lVert D_{\Lambda} \rVert^2 < 1$
\FOR{$\widetilde{z} = \left\lbrace z, z+\varepsilon\delta \right\rbrace$}
\STATE Choose $\widetilde{x}^{[0]}\in \mathbb{R}^N$, $x^{[0]}\in \mathbb{R}^N$, $y^{[0]} = D_{\Lambda} x^{[0]}$
\STATE $\partial_{\Lambda} \widetilde{x}^{[0]} \leftarrow 0_N$
\STATE $\partial_{\Lambda} x^{[0]} \leftarrow 0_N$
\STATE $\partial_{\Lambda} y^{[0]} \leftarrow D_{\Lambda} \partial_{\Lambda} x^{[0]}$
\FOR{$t = 0$ \TO $T_{\max}-1$}
\STATE \COMMENT{Primal-dual updates}
\STATE $y^{[t + 1]} = \mathrm{prox}_{\sigma_t \left(\lVert \cdot \rVert \right)^*} \left( y^{[t]} + \sigma_t D_{\Lambda} \widetilde{x}^{[t]}\right)$
\STATE $x^{[t + 1]} = \mathrm{prox}_{\tau_t \lVert A \cdot - \widetilde{z} \rVert^2} \left( x^{[t]} - \tau_t D_{\Lambda} y^{[t+1]}\right)$
\STATE $\vartheta_t = \sqrt{1 + 2\mu \tau_t}$, $\tau_{t+1} = \tau_t/\vartheta_t$ and $\sigma_{t+1} = \vartheta_t \sigma_t$
\STATE $\widetilde{x}^{[t+1]} = x^{[t+1]} + \vartheta_t\left( x^{[t+1]} - x^{[t]}\right)$
\STATE \COMMENT{Forward iterative differentiation}
\STATE $\partial_{\Lambda} y^{[t + 1]} = \partial_{y}\mathrm{prox}_{\sigma_t \left(\lVert \cdot \rVert \right)^*} \left( \partial_{\Lambda}y^{[t]} + \sigma_t D_{\Lambda} \partial_{\Lambda}\widetilde{x}^{[t]} + \sigma_t \left( \partial_{\Lambda}D_{\Lambda} \right) \widetilde{x}^{[t]}\right) $
\STATE $\partial_{\Lambda}x^{[t + 1]} = \partial_{x}\mathrm{prox}_{\tau_t \lVert A \cdot - \widetilde{z} \rVert} \left( \partial_{\Lambda}x^{[t]} - \tau_t D_{\Lambda} \partial_{\Lambda}y^{[t+1]} - \tau_t \left(\partial{\Lambda} D_{\Lambda}\right) y^{[t+1]}\right)$
\STATE $\partial_{\Lambda}\widetilde{x}^{[t+1]} = \partial_{\Lambda}x^{[t+1]} + \vartheta_t\left( \partial_{\Lambda}x^{[t+1]} - \partial_{\Lambda}x^{[t]}\right)$
\ENDFOR
\STATE  $\widehat{x}(\widetilde{z} ; \Lambda) \leftarrow x^{[T_{\max}]}$ 
\STATE $\partial_{\Lambda} \widehat{x}(\widetilde{z} ; \Lambda) \leftarrow \partial_{\Lambda} x^{[T_{\max}]}$
\ENDFOR
\STATE Compute $\mathrm{SURE}_{\varepsilon, \delta}(\Lambda)$ injecting $\widehat{x}(z ; \Lambda)$ and $\widehat{x}(z+\varepsilon\delta ; \Lambda)$ in Formula~\eqref{eq:SURE}
\STATE Compute $\mathrm{SUGAR}_{\varepsilon, \delta}(\Lambda)$ injecting $\widehat{x}(z ; \Lambda)$, $\widehat{x}(z+\varepsilon\delta ; \Lambda)$, $\partial_{\Lambda} \widehat{x}(z ; \Lambda)$ and $\partial_{\Lambda} \widehat{x}(z+\varepsilon\delta ; \Lambda)$ in Formula~\eqref{eq:SUGAR}
\RETURN{$\mathrm{SURE}_{\varepsilon, \delta}(\Lambda)$ and $\mathrm{SUGAR}_{\varepsilon, \delta}(\Lambda)$}
\end{algorithmic}
\end{algorithm}

\begin{algorithm}[h!]
\caption{\label{alg:BFGS} Automated selection of hyperparameters minimizing quadratic risk.}
\begin{algorithmic}
\REQUIRE $\varepsilon > 0$, $\delta \sim\mathcal{N}(0, \mathrm{Id}) \in \mathbb{R}^M$
\ENSURE $\Lambda^{[0]} \in \left(\mathbb{R}_+\right)^K$, $H^{[0]} \in \mathbb{R}^{K\times K}$
\STATE $\mathrm{SUGAR}^{[0]} \leftarrow \mathrm{SUGAR}_{\varepsilon, \delta}(\Lambda^{[0]})$ computed from Algorithm~\ref{algo:pd_sugar}
\FOR{$j = 0$ \TO $ J_{\max}-1$}
\STATE $d^{[j]} = - H^{[j]} \mathrm{SUGAR}^{[j]}$
\STATE $\alpha^{[j]} \in \underset{\alpha \in \mathbb{R}}{\mathrm{Argmin}} \, \mathrm{SURE}_{\varepsilon,\delta}( \Lambda^{[j]} + \alpha d^{[j]})$, SURE computed from Algorithm~\ref{algo:pd_sugar}
\STATE $\Lambda^{[j+1]} = \Lambda^{[j]} + \alpha^{[j]} d^{[j]}$
\STATE $\mathrm{SUGAR}^{[j+1]} \leftarrow  \mathrm{SUGAR}_{\varepsilon, \delta}(\Lambda^{[j+1]}) $ computed from Algorithm~\ref{algo:pd_sugar}
\STATE $u^{[j]} =\mathrm{SUGAR}^{[j+1]}  - \mathrm{SUGAR}^{[j]}$
\STATE $H^{[j+1]}  =\left( \mathrm{Id}  - \frac{d^{[j]} \left( u^{[j]} \right)^\top }{\left( u^{[j]} \right)^\top d^{[j]} }\right)H^{[j]} \left( \mathrm{Id}  - \frac{u^{[j]} \left( d^{[j]} \right)^\top }{\left( u^{[j]} \right)^\top d^{[j]} }\right) + \alpha^{[j]}  \frac{d^{[j]} \left( d^{[j]}\right)^\top  }{  \left( u^{[j]} \right)^\top d^{[j]} }.
$
\ENDFOR
\RETURN $\widehat{\Lambda}_{\mathrm{BFGS}} = \Lambda^{[T_{\max}]}$
\end{algorithmic}
\end{algorithm}

\section{Nonlinear denoising in non linear physics: low confinement solid friction and porous media multiphase flows}
\label{s:res}

\subsection{Low confinement solid friction: Piecewise linear denoising.}
\label{subsec:exp_sig}

\noindent {\bf Context } --
Friction experiments aim at probing not only the characteristics of materials, but also the dynamics of systems involving surfaces in contact. In particular, they are paradigms for modeling and attempting to predict earthquake dynamics \cite{Marone98}. The classical solid friction experiment consists of towing a mass $m$ (so-called {\it slider}) over a substrate via a spring of stiffness $k$ pulled at velocity $V$ (see for instance Figure~2 in \cite{Colas19}). The signal representative of the slider dynamics is the force measured at the contact point between the spring and the slider. Among the different regimes described in solid friction, we can distinguish the stick-slip, characterized by a tooth-shaped signal alternating slow, linear rise and fast drops, the inertial regime, in which the signal becomes periodic and resembles a sine curve, and the continuous sliding regime, characterized by an almost constant signal superimposed with noise \cite{Baumberger06}. The appearance of creep, slow forward motion of the slider previous to a slip event, may also modify the signal shape. The challenge in such studies is to establish a {\it regime diagram} describing (and therefore, predicting) the system dynamics depending on the parameters ($m,k,V$). If the recognition of the different regimes is easy for large mass $m$, experiments with low confinement pressure, necessary to avoid surface wear, are challenging as they add noise to the experimental signals \cite{Colas19}. In this context, new signal processing tools are required. Here we focus in particular on signal denoising by approximating, at first order, the stick-slip signals to piecewise linear signals.\\

\noindent {\bf Data } --
Experiments of solid friction (taken from \cite{Colas19}) were performed by pulling a mass $m=30.7$~g (slider area $9 \times 6$~cm$^2$)  over a solid substrate. Both surfaces in contact consist of paper samples (Canson$^\circledR$, characterized by its roughness). A cantilever spring (metallic blade of stiffness $k$ between 168 and 3337 N/m) is pulled at constant velocity $V$ (between 42 and 7200~$\mu$m/s) and is in contact with the slider by a steel ball glued onto this latter, ensuring a punctual contact and the free motion of the contact point. The slider dynamics is quantified though the measurement of the blade deflection, $\Delta x$, by an inductive sensor (Baumer, IPRM 12I9505/S14). In all experiments, the mass $m$ is kept constant. We vary the parameters ($k,V$) and, for each experiment, record the normalized force signal $F^*$ from the blade deflection, $F^*= k \Delta x / (mg)$, where $g=9.81$~m/s$^{-2}$ is the gravitational acceleration. This signal is recorded with a sampling frequency of 2~kHz, and its size varies from about $4.5 \times 10^3$ to $7.5 \times 10^5$. \\

\noindent \textbf{Piecewise linear denoising} -- 
In~\cite{Colas19}, stick-slip signals were processed using an optimization formalism, falling under formulation~\eqref{eq:nl}, in order to enforce piecewise linear behavior.
The observation $z$ corresponds to the measured force signals, $A = \mathrm{Id}$, the linear operator $D_{\Lambda}$ is chosen to be the discrete Laplacian described in Equation~\eqref{eq:l12} (for an univariate signal, i.e. $K=1$), and $\lVert \cdot \rVert$ is the $\ell_{1,2}$-norm defined in~\eqref{eq:l12norm}, which reduces to the $\ell_1$-norm in the context of univariate signals.\\
The tedious task of tuning the regularization parameter $\lambda$ was performed by expert visual inspection and led to a choice $\lambda_{\mathrm{expert}} = 0.8$ uniformly applied to all signals, irrespective of the different experiment settings. 
Examples of noisy observations are shown in Figure~\ref{fig:stick_slip}(gray), with associated piecewise linear estimates obtained with $\lambda_{\mathrm{expert}}$ displayed in red. 
The regularized signals appear to capture well the transition between the stick and slip regimes.

Here, we propose to illustrate the use of the regularization parameter automated tuning strategy presented in Section~\ref{sec:selec_lambda} for piecewise linear denoising on stick-slip signals.
$\mathrm{SURE}_{\varepsilon, \delta}(\lambda)$ is used as the quality criterion and minimized over $\lambda$.
Therefore, both grid search and automated tuning are implemented and compared.\\

\begin{figure}[h!]
\centering
\includegraphics[trim = 3cm 3.5cm 3.5cm 2.25cm, clip, width = \linewidth]{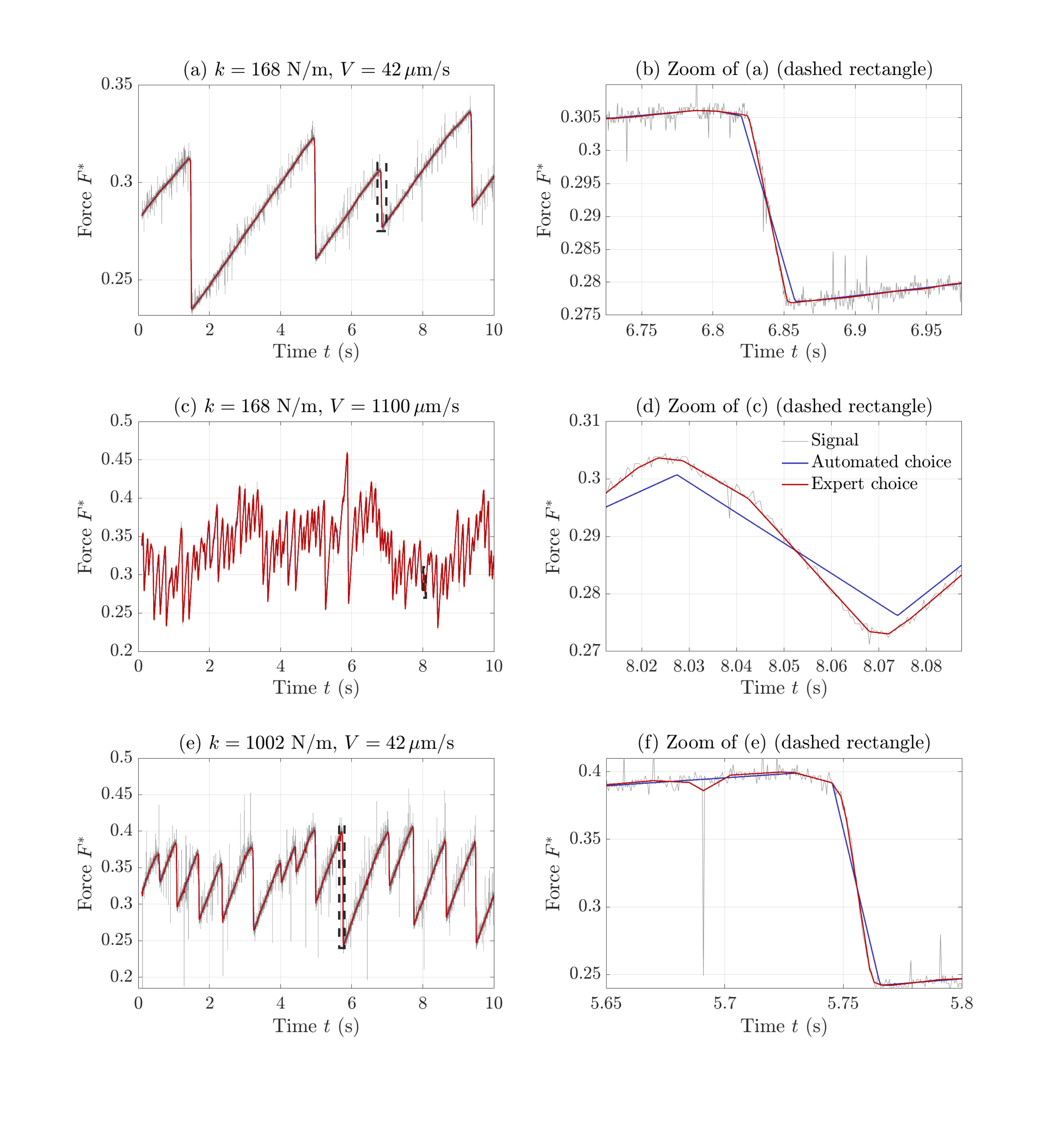}
\caption{\label{fig:stick_slip}
{\bf Stick-slip force signals.} Experimental data ((in grey) for three different experimental settings and nonlinear filtering enforcing piecewise linear behavior, with automated hyperparameter tuning (blue) and expert-selected hyperparameter (red).
}
\end{figure}

\noindent \textbf{Automated data-driven hyperparameter tuning} --
The Finite Difference step $\varepsilon$, involved in $\mathrm{SURE}_{\varepsilon, \delta}$ and $\mathrm{SUGAR}_{\varepsilon, \delta}$ computation [see Equations~\eqref{eq:SURE} and \eqref{eq:SUGAR}] is set to
\begin{align}
\varepsilon =  \frac{2\sigma}{N_1^{0.3}}
\end{align}
with $N_1$ the length of the considered stick-slip signal and $\sigma^2$ the estimated variance of the additive noise corrupting the signal.
Since no additional information about the noise is available, $\sigma^2$ is estimated using the sample variance estimator applied to observations.\\
$\mathrm{SURE}_{\varepsilon, \delta}$ (black curve in Figure~\ref{fig:grid_bfgs_plots}) is first computed over a grid of 15 logarithmically spaced values of the regularization parameter $\lambda$, using Algorithm~\ref{alg:grid}.
Then, $\lambda_{\mathrm{grid}}$ is defined as the minimum of $\mathrm{SURE}_{\varepsilon, \delta}(\lambda)$ over the grid and indicated by the `+' symbol.
Finally, the quasi-Newton Algorithm~\ref{alg:BFGS} for automated tuning of regularization parameter is run, providing $\lambda_{\mathrm{BFGS}}$, represented by the `$\color{bleu}\ast$' symbol.
The regularization parameter chosen by the expert is displayed for comparison purpose, an indicated by the `$\color{rouge}\times$' marker.
For each experimental setting $(k,V)$, the grid search optimal regularization parameter $\lambda_{\mathrm{grid}}$ and the automatically tuned regularization parameter $\lambda_{\mathrm{BFGS}}(k,V)$ obtained respectively from Algorithms~\ref{alg:grid}~and~\ref{alg:BFGS} are compared in Table~\ref{tab:lambda_val}, showing satisfactory agreement.\\

\begin{table}[h!]
\centering
\begin{tabular}{ccccc}
\toprule
 \multicolumn{2}{c}{$\lambda_{\mathrm{grid}}(k,V)$}& \multicolumn{3}{c}{$k$ (N/m)}\\
 \midrule
 &  & \textbf{168} & \textbf{1002}& \textbf{2254}\\
  \midrule
\multirow{3}{*}{\rotatebox{90}{$V$ ($\mu$m/s)}} &\textbf{42} & 21.6  & 12.7  & 23.1\\
\noalign{\vspace{2mm}}
& \textbf{1100} &  16.6 & 3.7 & 76.6 \\
\noalign{\vspace{2mm}}
& \textbf{4300} & 8.8 & 6.2 & 2.6 \\
\bottomrule
\end{tabular} \hspace{1cm}
\begin{tabular}{ccccc}
\toprule
 \multicolumn{2}{c}{$\lambda_{\mathrm{BFGS}}(k,V)$}& \multicolumn{3}{c}{$k$ (N/m)}\\
 \midrule
 &  & \textbf{168} & \textbf{1002}& \textbf{2254}\\
  \midrule
\multirow{3}{*}{\rotatebox{90}{$V$ ($\mu$m/s)}} &\textbf{42} &  7.9 & 10.0 & 0.2\\
\noalign{\vspace{2mm}}
&\textbf{1100} & 16.5 & 3.4 & 2.2\\
\noalign{\vspace{2mm}}
&\textbf{4300} & 9.1 & 4.9 & 3.2\\
\bottomrule
\end{tabular}
\caption{\label{tab:lambda_val}Grid search v.s. automated tuning of regularization parameter in piecewise linear denoising of stick-slip signals for different stiffness $k$ and velocity $V$.}
\end{table}

\begin{figure}[h!]
\centering
\includegraphics[width = \linewidth]{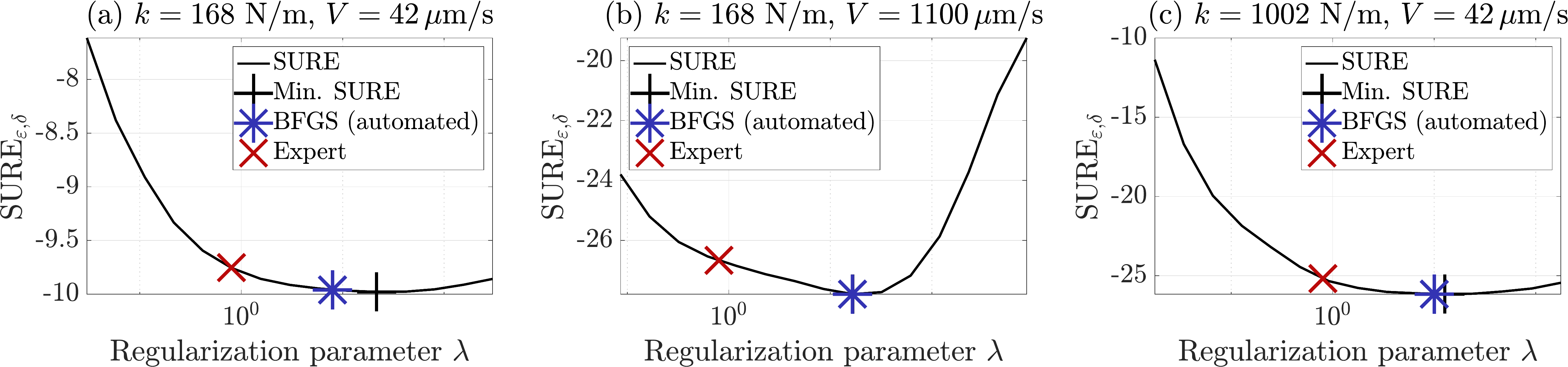}
\caption{\label{fig:grid_bfgs_plots}$\mathrm{SURE}_{\varepsilon, \delta}(\lambda)$. Grid search v.s. automated tuning of regularization parameter.}
\end{figure}

\noindent \textbf{Denoised data analysis} -- Table~\ref{tab:lambda_val} shows first that $\lambda_{\mathrm{BFGS}}$ is within the same order of magnitude as $\lambda_{\mathrm{expert}}$.
This is consistent with Figure~\ref{fig:stick_slip}, that further shows that denoised experimental signals obtained from nonlinear filtering enforcing piecewise linear behavior, with automated hyperparameter tuning (blue) and expert-selected hyperparameter (red) display similar shapes and behaviors. 
This is a very satisfactory outcome as the proposed data-driven and automated hyperparameter tuning yields outcomes very consistent with those obtained from expert choices, without making use of any a priori information, and relying on data only instead.

Table~\ref{tab:lambda_val} also shows that the automated procedure yields different regularization parameters for the different $(k,V)$ configurations, illustrating
an ability to finely adapt to data, which would not be possible - or would be too much time-consuming - for an expert.
Table~\ref{tab:lambda_val} further reveals that the automatically selected regularization parameters, $\lambda_{\mathrm{BFGS}} $, are, for almost all $(k,V)$ configurations, slightly larger than the expert-selected ones,  $ \lambda_{\mathrm{expert}}$, hence yielding overall more regular signals.
Figure~\ref{fig:stick_slip}(f) shows that the red signal, obtained with $\lambda_{\mathrm{expert}}$, displays discontinuities (e.g, around $t =5.7$~s) which are due to noise rather than to the physical mechanisms of interest, that are satisfactorily properly discarded on blue signal, obtained with the automated selection $\lambda_{\mathrm{BFGS}} $, hence showing the interest of tuning the hyperparameter to each specific signal.
However, Figure~\ref{fig:stick_slip}(b) and (d) also shows small yet visible differences during the slip-phase (fast decrease) between the red signals, obtained with $\lambda_{\mathrm{expert}}$, and the blue signals, obtained with the automated selection $\lambda_{\mathrm{BFGS}} $. 
To decide which one is the most relevant requires returning to a detailed analysis of solid friction: 
The stick phases actually produce force signals that are exactly linearly increasing~; 
For the slip phase, while they can be described in first approximation as an abrupt linear decrease, detailed analysis indicates that they actually consist of arches of sinusoidal functions that connect the stick phases. 
Therefore, it can be considered that the expert-driven signals (red) better fit the experimental data, at the price though of concatenating a series of short linear segments that are irrelevant with respect to the underlying physics, whereas the data-driven signals (blue) yield more stylized piecewise linear approximations of the data, that may however better capture the times of transitions between stick and slip phases, an information of premier importance to analyze solid friction regimes.

In sums, deciding between the use of expert versus automated tuning of the hyperparameters combines several issues ranging from feasibility (expert tuning is time consuming, prone to errors and may lack reproducibility) to relevance (denoised signals must permit relevant access to quantity of interest for the physics). \\

\subsection{Porous media multiphase flow: Piecewise homogeneous texture segmentation with weighted isotropic TV.}
\label{ss:textseg}

\noindent \textbf{Context} -- Understanding and predicting the dynamics of multiphase flows is a major issue in geosciences (soil decontamination, CO$_2$ sequestration) and in the industry (enhanced oil recovery, heterogeneous catalysis) \cite{reddy2001effects,Kang05,hessel05catalysis,Babchin08}. 
Among these processes, many involve a joint gas and liquid flow through a porous medium. 
Quantifying the contact areas between the different phases, where chemical reactions take place, is of tremendous importance for analyzing and predicting the efficiency of such processes \cite{kreutzer05chemreaction}. 
However, even when direct visualization is possible, the porous medium generates a global, multiscale texture on images which makes it difficult to extract the gas-liquid interfaces. 
Segmentation techniques based on morphological tools used so far to differentiate phases in multiphase flows \cite{serres2016stability} present severe limitations: arbitrary threshold setting, non-physical irregular bubble contour, non detection of small bubbles. 
In addition, recent developments in high-resolution and high-speed imaging yield large-size images and large data sets, thus bringing forward issues in memory and computational costs.
Here, we focus on the identification of the different phases (liquid and gas) in textured images.
As a first approximation, the liquid and the gas appear as homogeneous fractal textures.
Hence, discriminating phases requires to solve a texture segmentation problem. \\

\noindent \textbf{Data} --  Experiments of joint gas and liquid flow through a porous medium were performed in a quasi-2D vertical Hele-Shaw cell of width 210~mm, height 410~mm and gap 1.75~mm (see Figure~1 in \cite{Busser20}). The porous medium is an open cell solid foam of NiCrFeAl alloy (Alantum), with a typical pore diameter of 580~$\mu$m. Constant gas and liquid flow rates are injected at the bottom of the cell through nine injectors (air) and a homogeneous slit (water). Images of the multiphase flow are acquired by a high-resolution camera (Basler A2040-90um, $2048 \times 2048$~pixels + 16~mm lens) at 100~Hz  \cite{Serres_M_2017_PhD,Busser20}. 
After cropping the region of interest, the size of the images to analyze is $1626\times1160$. 
An example is provided in Figure~\ref{fig:txt_seg}(a), showing that the gas phase (dark gray or white structures) is textured because of the presence of the foam struts which are not captured by the camera resolution. 
The liquid phase (in gray) is also textured though at smaller scales, as can be observed in Figure~\ref{fig:txt_seg_zoom}(a).
For all experimental data sets, 50 to 3000 images are recorded. Similarly to the friction experiment, a large number of data sets associated with different parameters (here the gas and liquid flow rate) are investigated, to analyze the different hydrodynamic regimes.\\

\noindent \textbf{Fractal features} -- We consider \textit{fractal}, or \textit{scale-free}, features,  consisting of the local behavior as functions of scales of the wavelet \textit{leader} coefficients $\mathcal{L}_{j,\underline{n}}$ and scale $j\in\{1,\ldots, J\}$, built as a local supremum of wavelet coefficients~\cite{wendt2007bootstrap,Wendt2009b}.
For each pixel $\underline{n} \in \Omega = \left\lbrace 1, \hdots, N_1\right\rbrace \times\left\lbrace 1, \hdots, N_2\right\rbrace$, the \textit{leader} coefficients of the image $X$ to analyze, denoted $ \mathcal{L}_{j,\underline{n}}$, evidence the following local scaling property~\cite{jaffard2004wavelet}
\begin{align}
\label{eq:Ljn}
 \mathcal{L}_{j,\underline{n}}\sim \eta_{\underline{n}} 2^{jh_{\underline{n}}}, \quad  \text{as} \, 2^j \rightarrow 0
\end{align}
where $2^j$ denotes the scale of the multiscale transform.
The quantity $h_{\underline{n}}$ measures the \textit{local regularity} of the texture at pixel $\underline{n}$.
In log-log coordinates, Equation~\eqref{eq:Ljn} corresponds to a linear behavior through octaves $j$
\begin{align}
\label{eq:ljn}
\log_2(\mathcal{L}_{j,\underline{n}}) \simeq \log_2(\eta_{\underline{n}}) + j h_{\underline{n}} ,\quad \text{as} \, 2^j \rightarrow 0.
\end{align}
Setting $v_{\underline{n}} := \log_2(\eta_{\underline{n}})$, which will be called in the following the \textit{local power} of the texture, a texture $X$ is characterized by $\left( h_{\underline{n}}, v_{\underline{n}} \right)_{\underline{n} \in \Omega}$.
\begin{definition}
An \textit{homogeneous} texture is characterized by a uniform local regularity $h_{\underline{n}} \equiv H$ and local power $v_{\underline{n}} \equiv V$.
\end{definition}
Then, texture segmentation consists in identifying a partition of the image domain 
\begin{align}
\Omega = \Omega_1 \cup \cdots \cup \Omega_Q, \quad \Omega_q \cap \Omega_{q'} = \emptyset \text{ for } q \neq q',
\end{align}
for which both $h_{\underline{n}}$ and $v_{\underline{n}}$ are uniform on each $\Omega_q$.
In other words, it consists in obtaining piecewise constant maps of local regularity and local power.\\

\noindent \textbf{Regularized estimates} -- Linear regression on log-\textit{leaders}~\eqref{eq:ljn} can be formulated as the minimization of the following least-squares
\begin{align}
\label{eq:defF}
\Phi( h, v ; \mathcal{L}) = \frac{1}{2} \sum_{j = j_1}^{j_2}  \left\lVert  j  h + v - \log_2  \mathcal{L}_{j}\right\rVert^2,
\end{align}
and provides estimates $\left(\widehat{h}_{\mathrm{LR}}, \widehat{v}_{\mathrm{LR}}\right)$ of fractal features.
As an example, the linear regression estimate of the local regularity of the (zoomed) flow image of Figure~\ref{fig:txt_seg}(a) (Figure~\ref{fig:txt_seg_zoom}(a)) is presented in Figure~\ref{fig:txt_seg}(b) (Figure~\ref{fig:txt_seg_zoom}(b)).
These estimates turn out to suffer from large variances precluding their use of actual segmentation,
thus calling for nonlinear estimation tools.\\
To favor piecewise homogeneous segmentation, we enforce piecewise constancy in estimated features via two different Total Variation-based penalizations, leading to the minimization of the \textit{Joint} and the \textit{Coupled} functionals
\begin{align}
\label{eq:penal}
\left(\widehat{h}^{\mathrm{J/C}}, \widehat{v}^{\mathrm{J/C}}\right) = \underset{h, v}{\arg\min} \, \,  \Phi(h,v ; \mathcal{L}) + \lambda \Psi_{\mathrm{J/C}}(h,v ; \alpha).
\end{align}
The \textit{Joint} and \textit{Coupled} penalizations are defined as
\begin{align}
\label{eq:penj}\Psi_{\mathrm{J}}(h,v) &:= \lambda \left( \alpha \mathrm{TV}(h) + \mathrm{TV}(v) \right), \\\
\label{eq:penc}\Psi_{\mathrm{C}}(h,v) &:= \lambda  \sum_{n_1 = 1}^{N_1-1} \sum_{n_2 = 1}^{N_2-1} \sqrt{\alpha^2 \left(Hh\right)_{n_1,n_2}^2 +\alpha^2 \left(Vh\right)_{n_1,n_2} ^2 +  \left(Hv\right)_{n_1,n_2}^2 +\left(Vv\right)_{n_1,n_2} ^2 } . 
\end{align}
where the total variation (TV) is defined in Equation~\eqref{eq:l12_2D} and the horizontal and vertical discrete gradients, $H$ and $V$, are defined at Equation~\eqref{eq:diff_2D}.
While the \textit{Joint} penalization imposes independently piecewise constancy of local regularity $h$ and local power $v$, the \textit{Coupled} penalization is more restrictive and favors co-localized changes in $h$ and $v$.
The trade-off between fidelity to the mathematical model~\eqref{eq:ljn} and piecewise constancy of $h$ and $v$ is controlled by the regularization parameter $\lambda > 0$ and $\alpha  > 0$.\\

\begin{figure}[h!]
\centering
\begin{tabular}{cc}
(a)~Flow image & (b)~Linear regression \\
\includegraphics[width = 3cm]{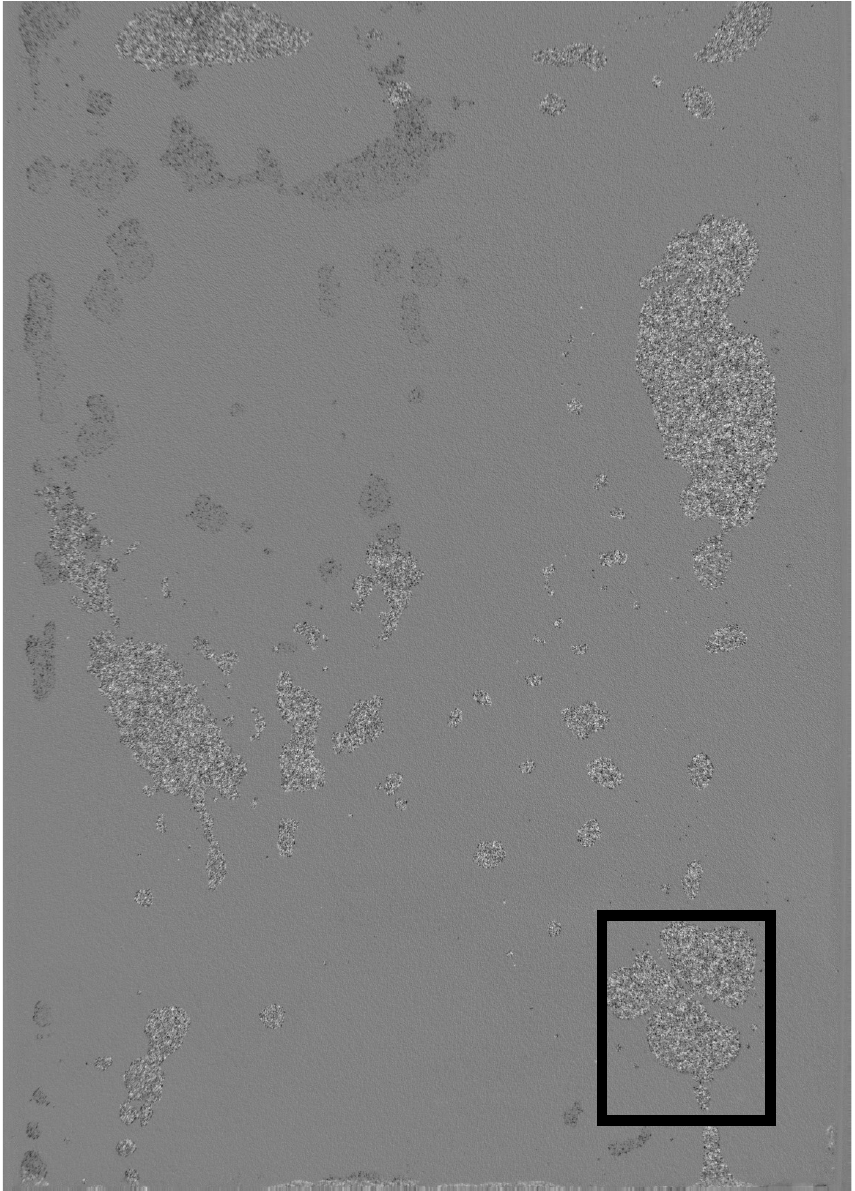}  &  \includegraphics[width = 3cm]{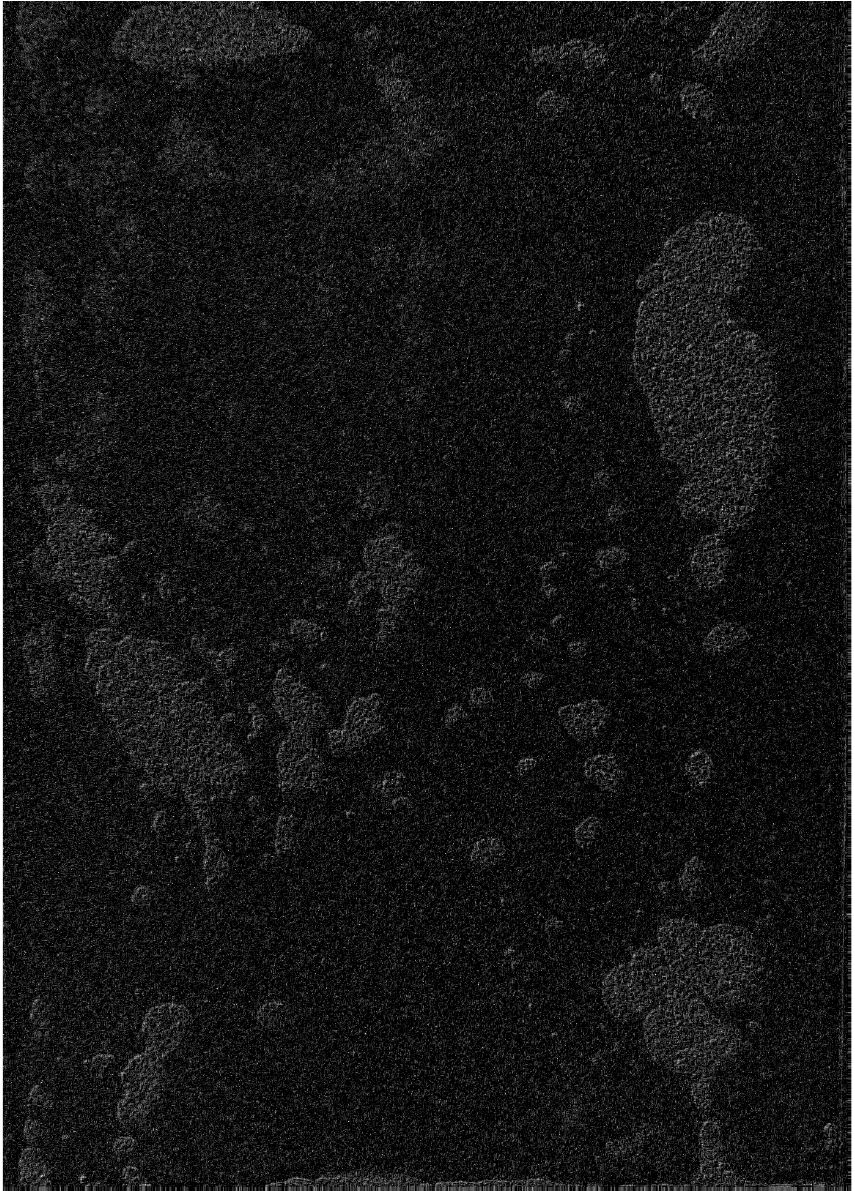}\\
(c)~T-ROF-$\mathrm{Id}$ & (d)~ROF-$\mathrm{Id}$ \\
\includegraphics[width = 3cm]{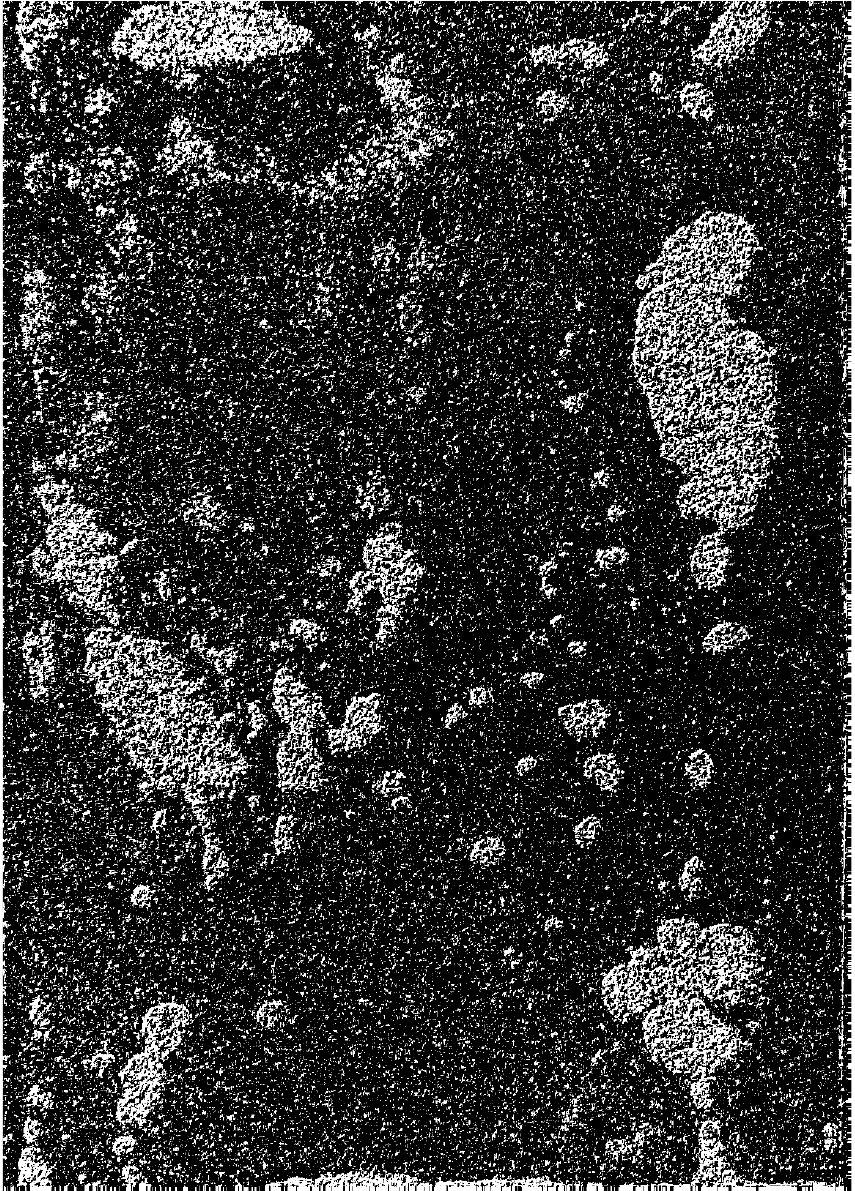} &  \includegraphics[width = 3cm]{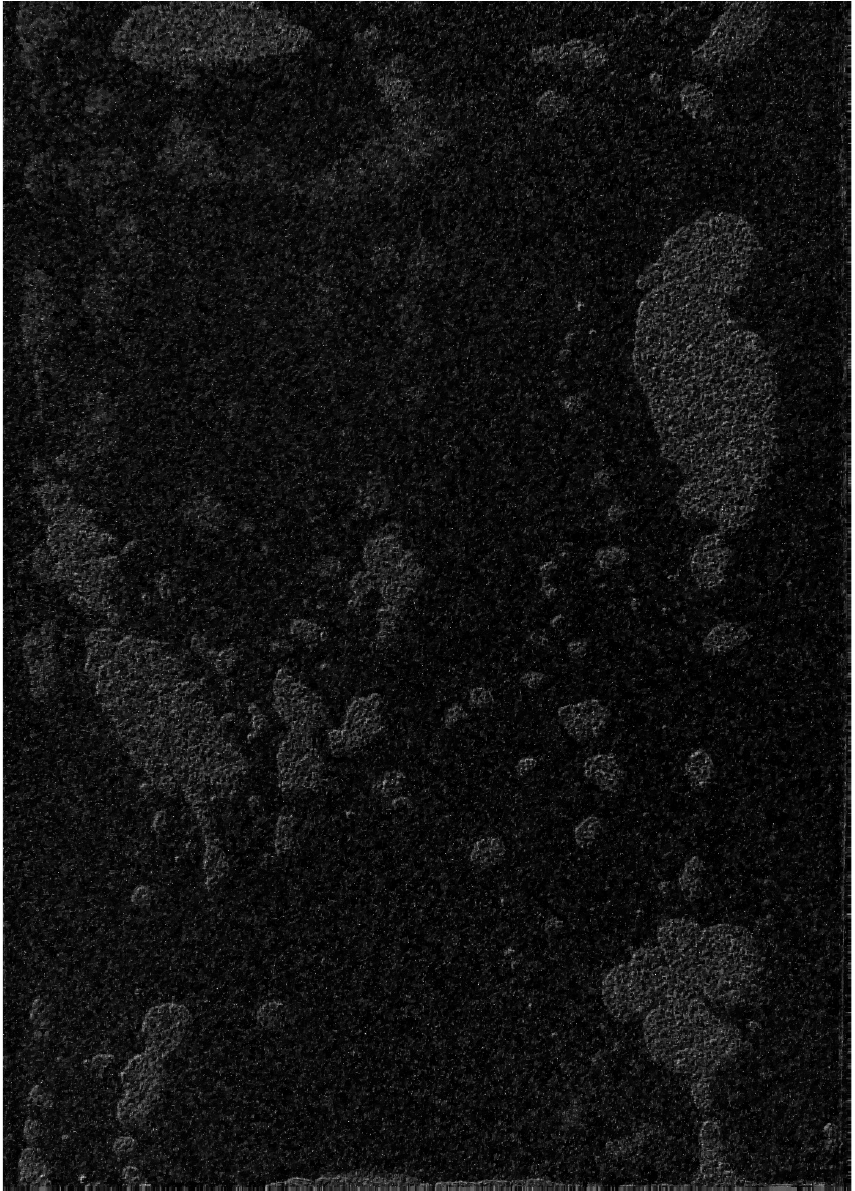} \\
(e)~T-ROF-$\mathcal{S}$ & (f)~ROF-$\mathcal{S}$\\
 \includegraphics[width = 3cm]{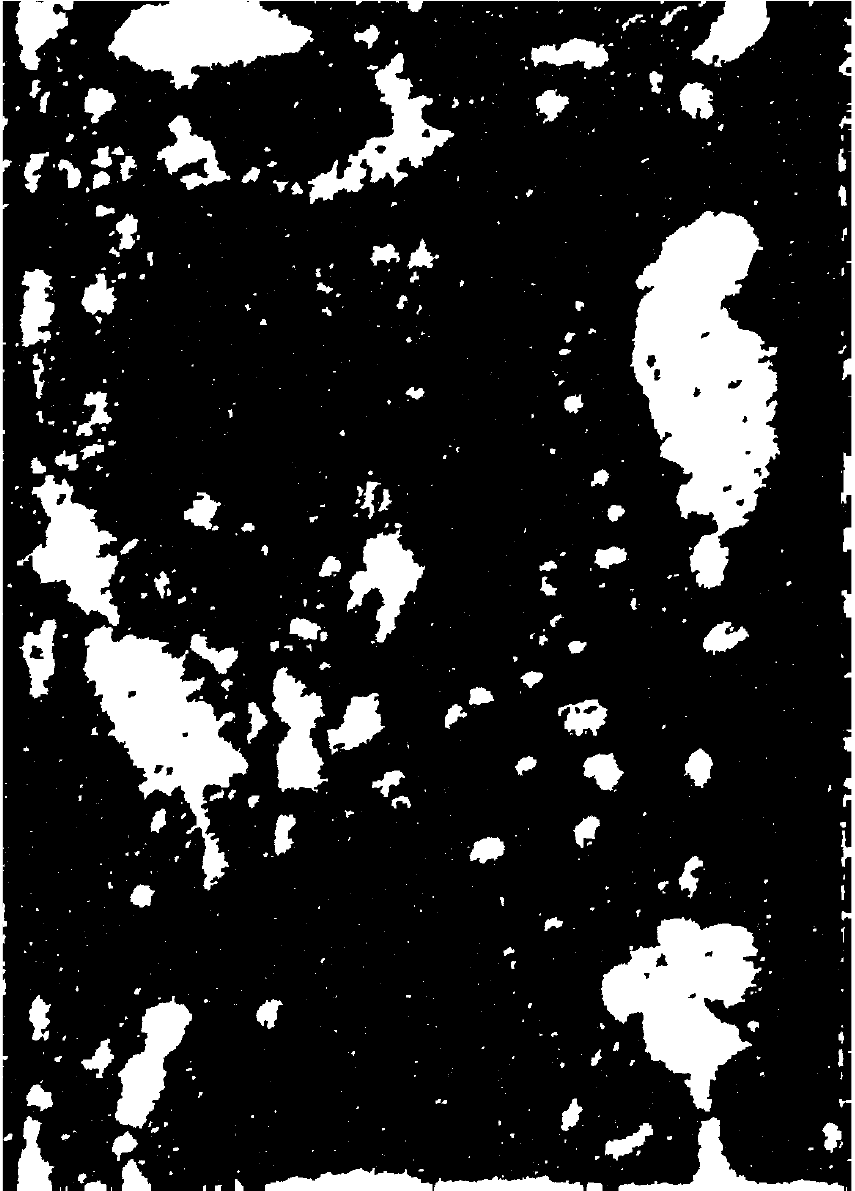} & \includegraphics[width = 3cm]{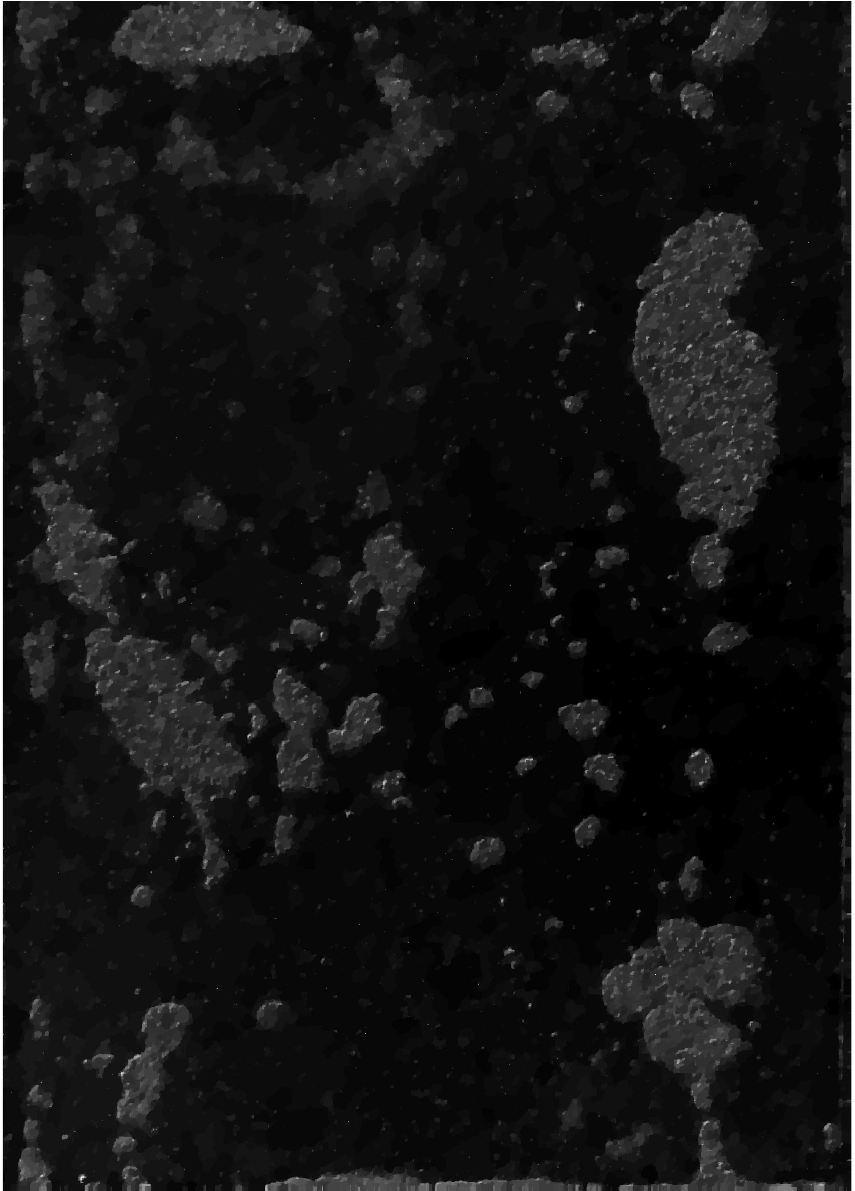} \\
(g)~T-\textit{Joint} & (h)~\textit{Joint} \\
\includegraphics[width = 3cm]{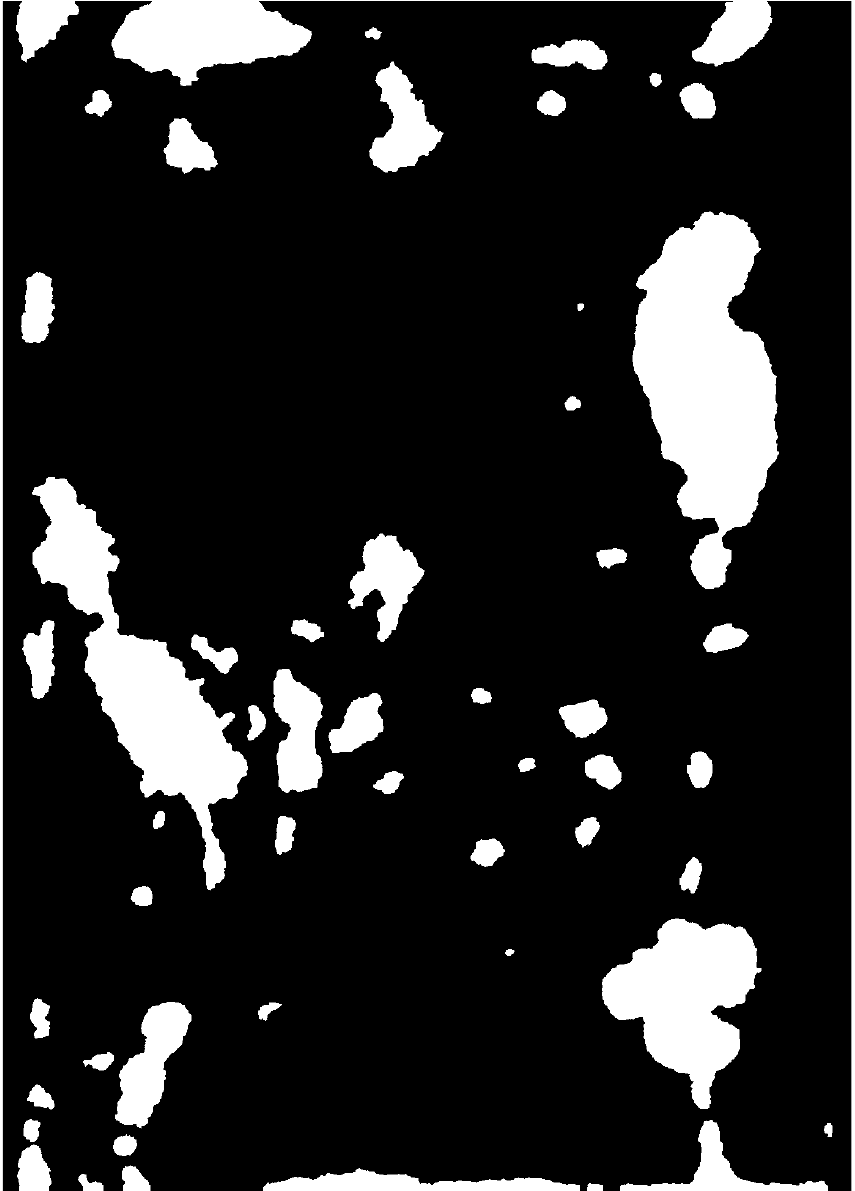} & \includegraphics[width = 3cm]{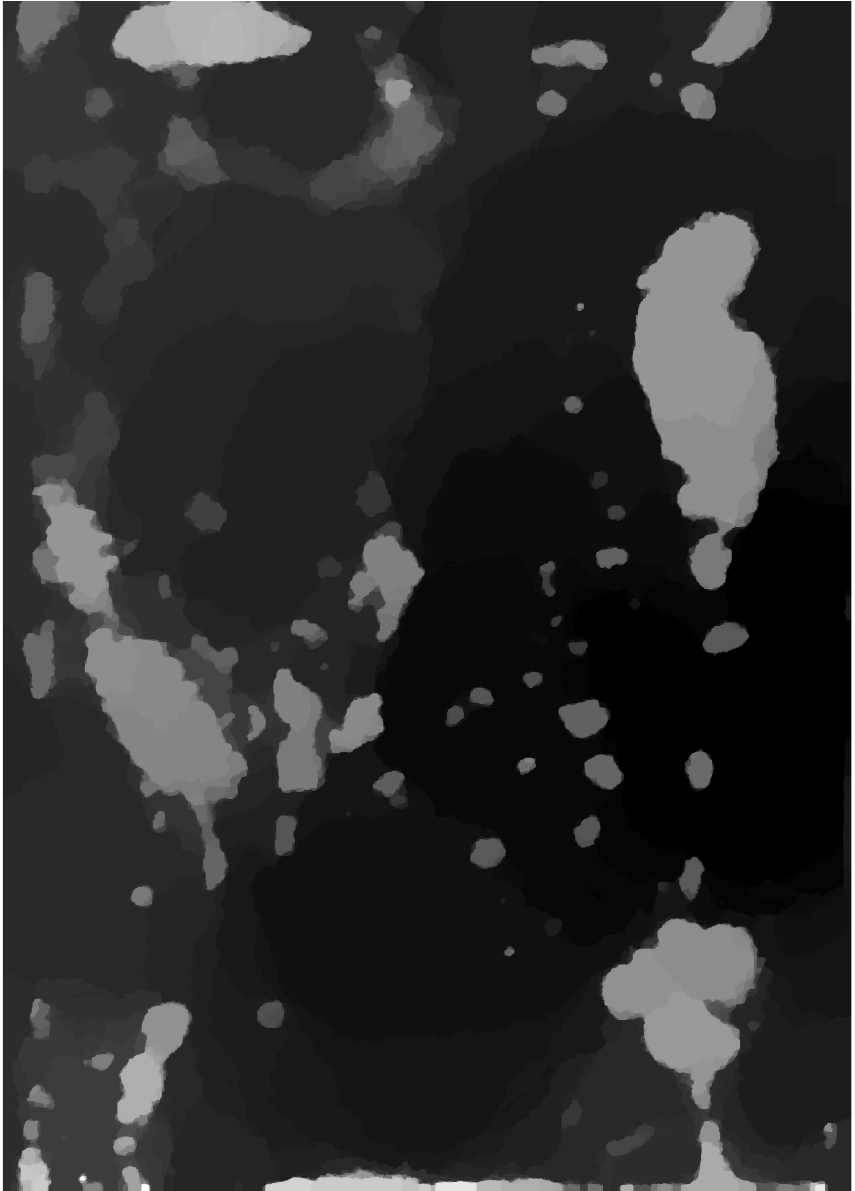} \\
(i)~T-\textit{Coupled} & (j)~\textit{Coupled} \\
\includegraphics[width = 3cm]{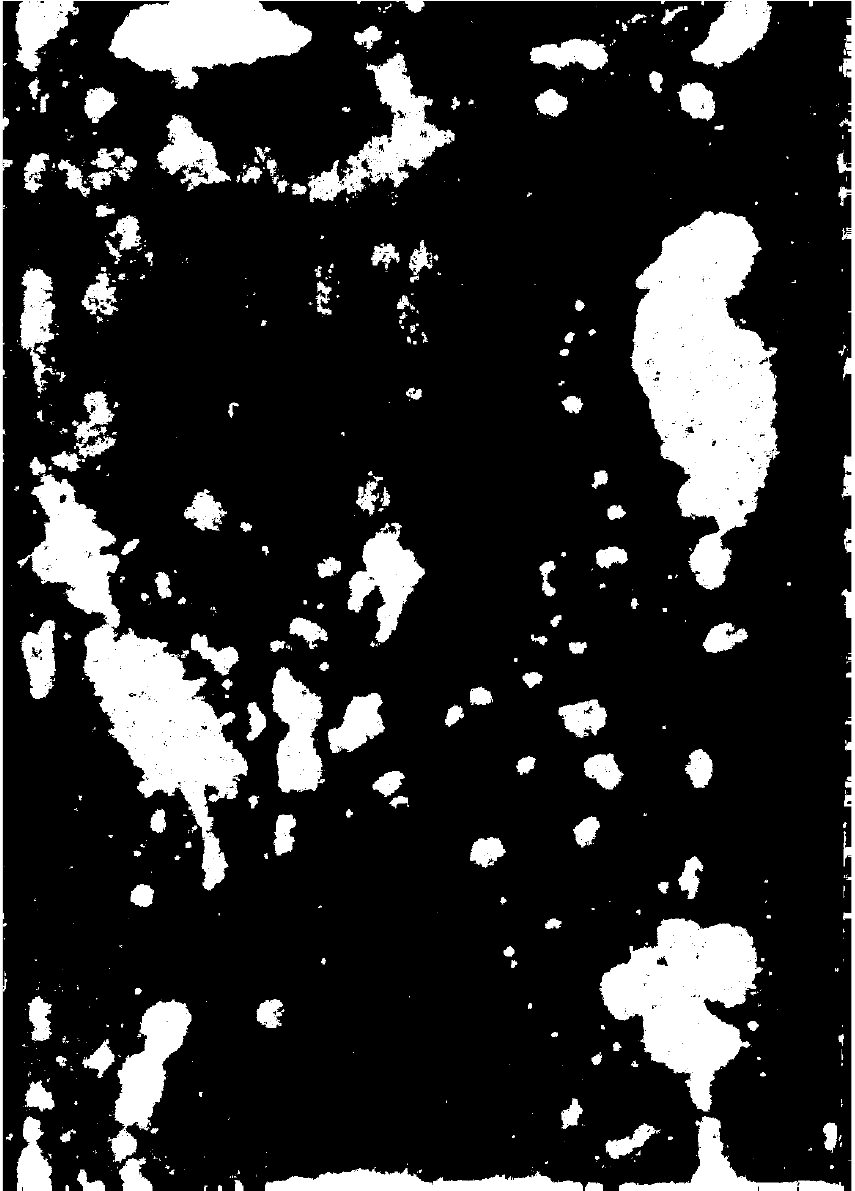} & \includegraphics[width = 3cm]{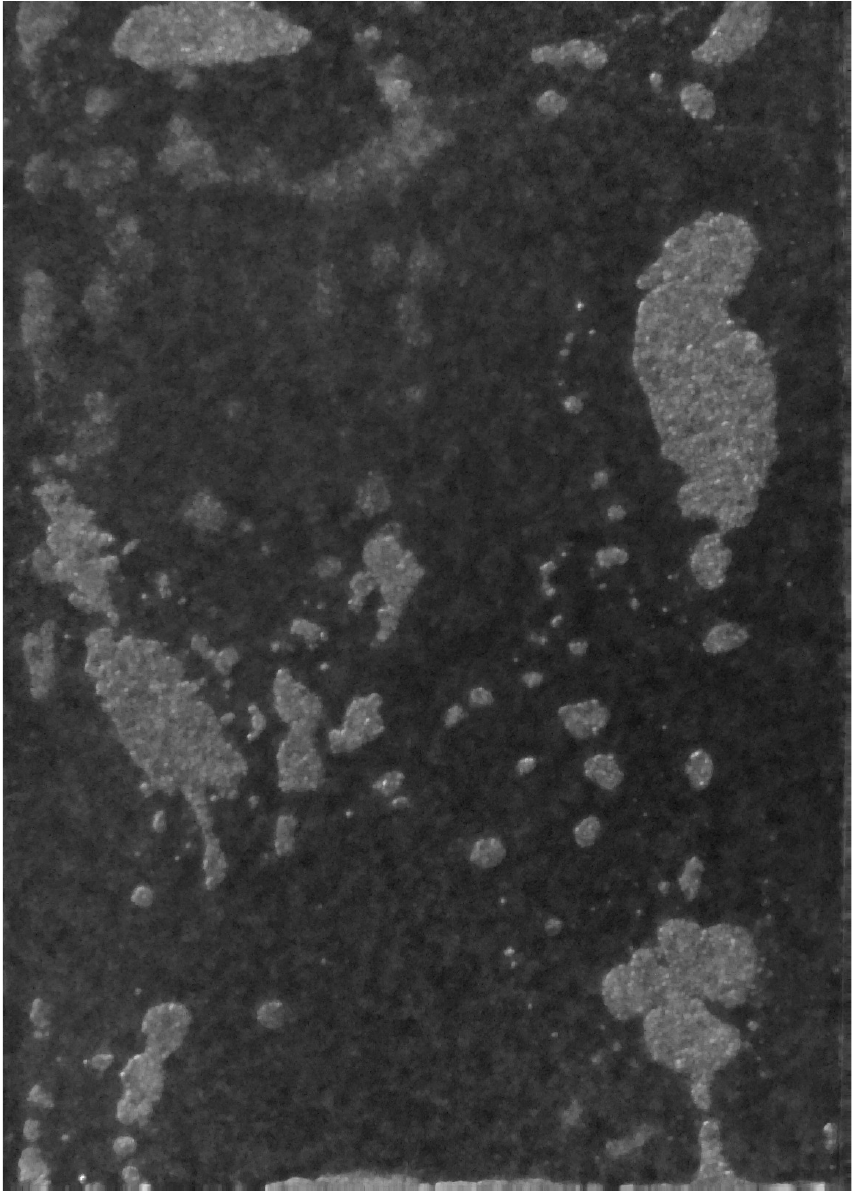}
\end{tabular}
\caption{\label{fig:txt_seg}{\bf Porous media multiphase flow texture segmentation based on fractal features.} Comparisons between different approaches as summarized in Table~\ref{tab:meths}.}
\end{figure}

\begin{figure}[h!]
\centering
\begin{tabular}{cc}
(a)~Zoomed flow image & (b)~Linear regression \\
\includegraphics[width = 3cm]{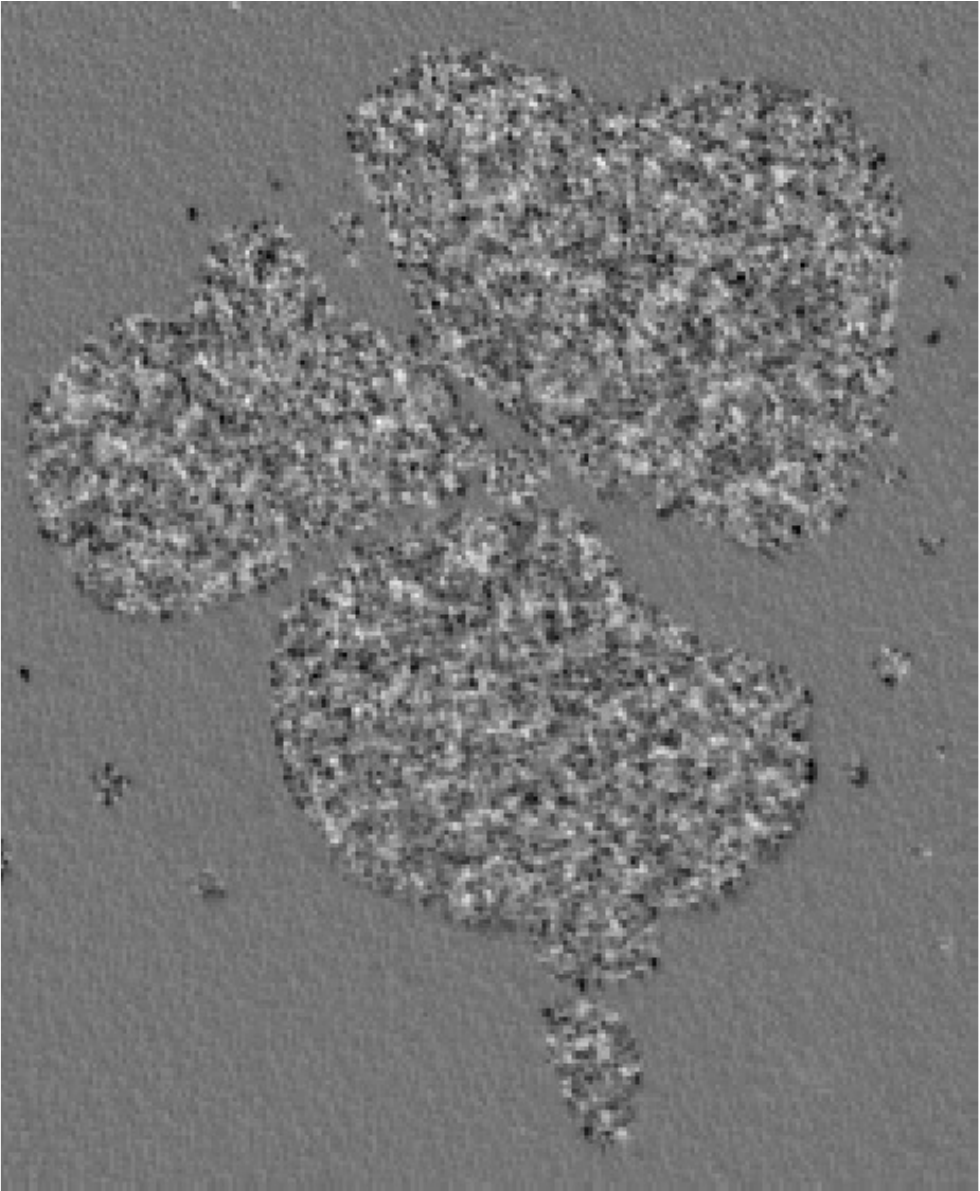}  &  \includegraphics[width = 3cm]{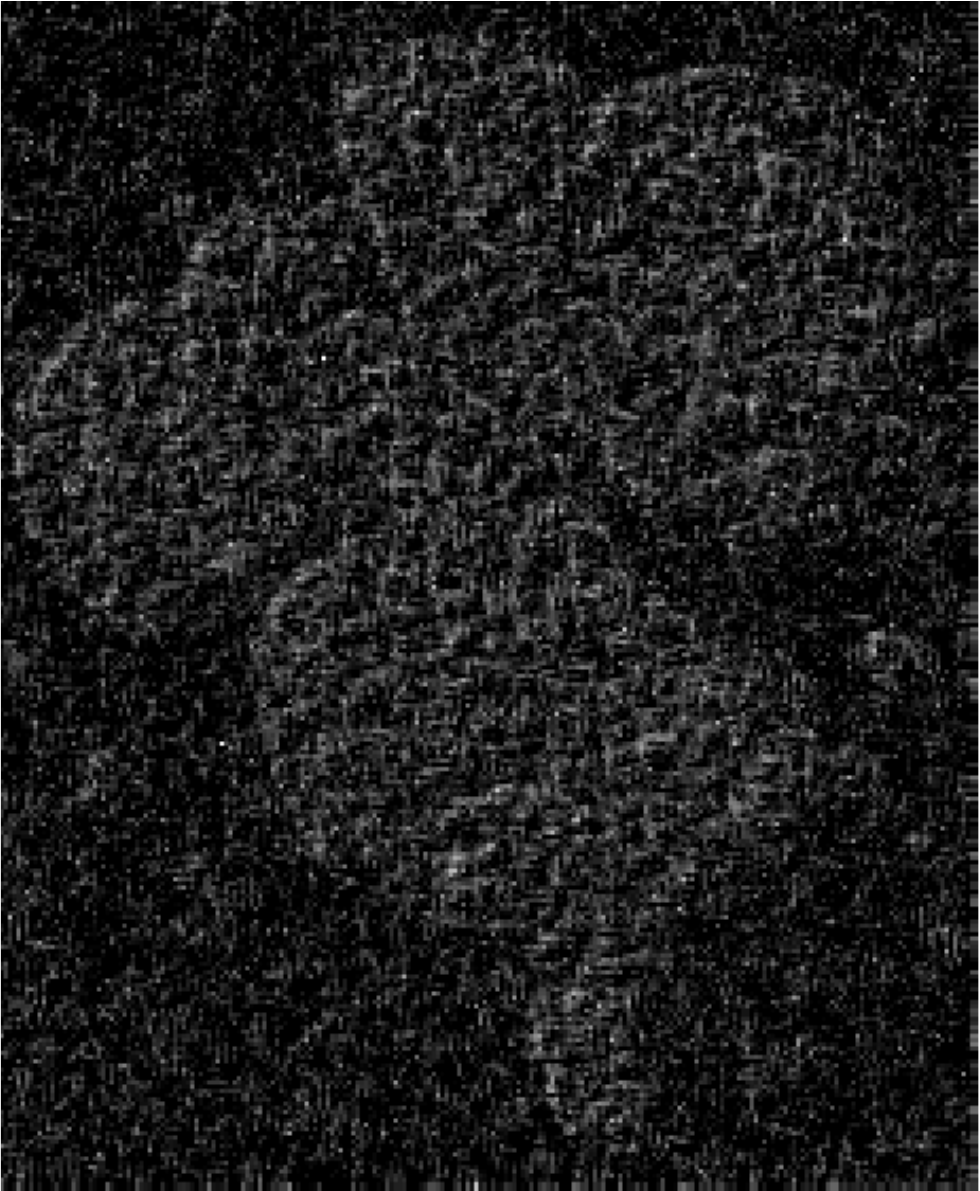}\\
(c)~T-ROF-$\mathrm{Id}$ & (d)~ROF-$\mathrm{Id}$ \\
\includegraphics[width = 3cm]{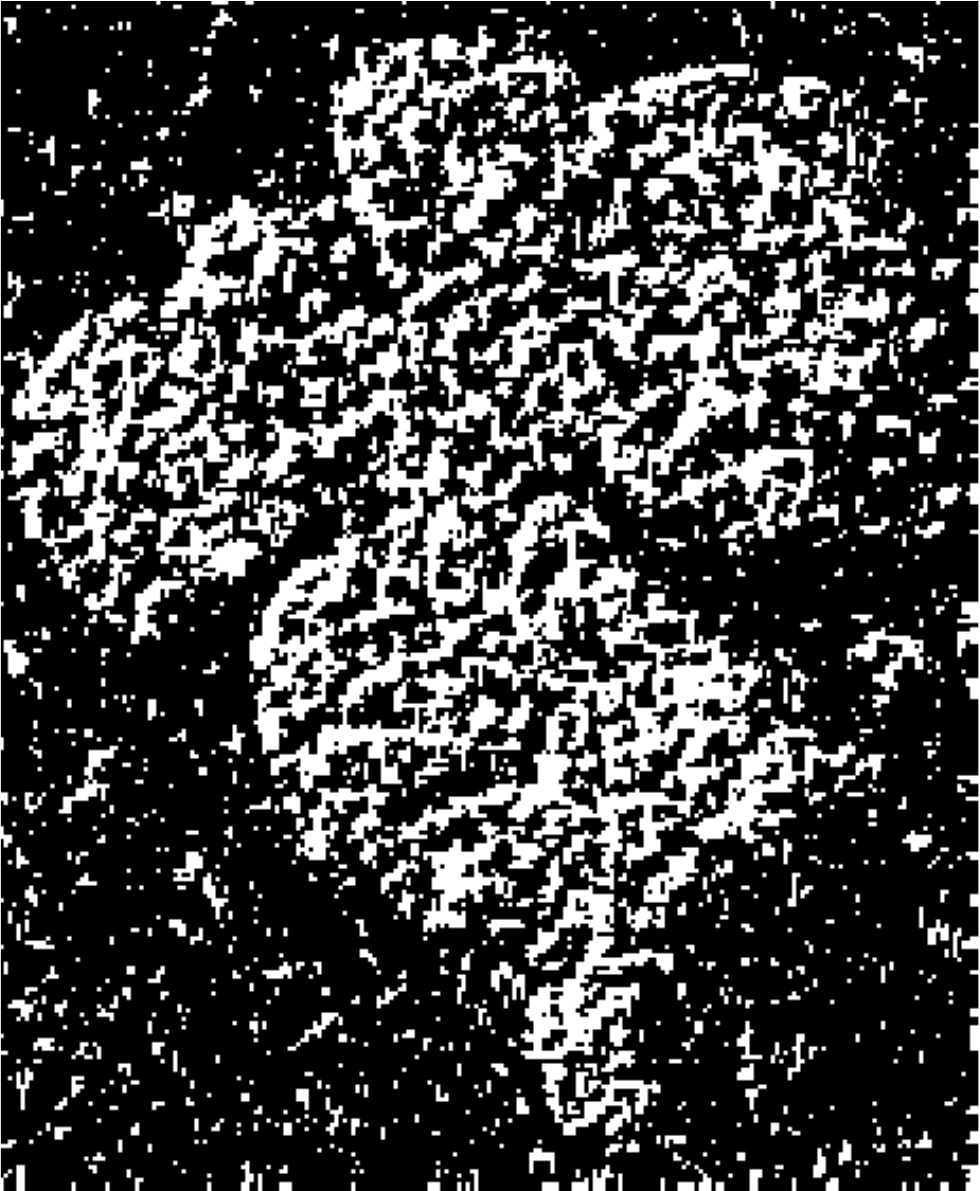} &  \includegraphics[width = 3cm]{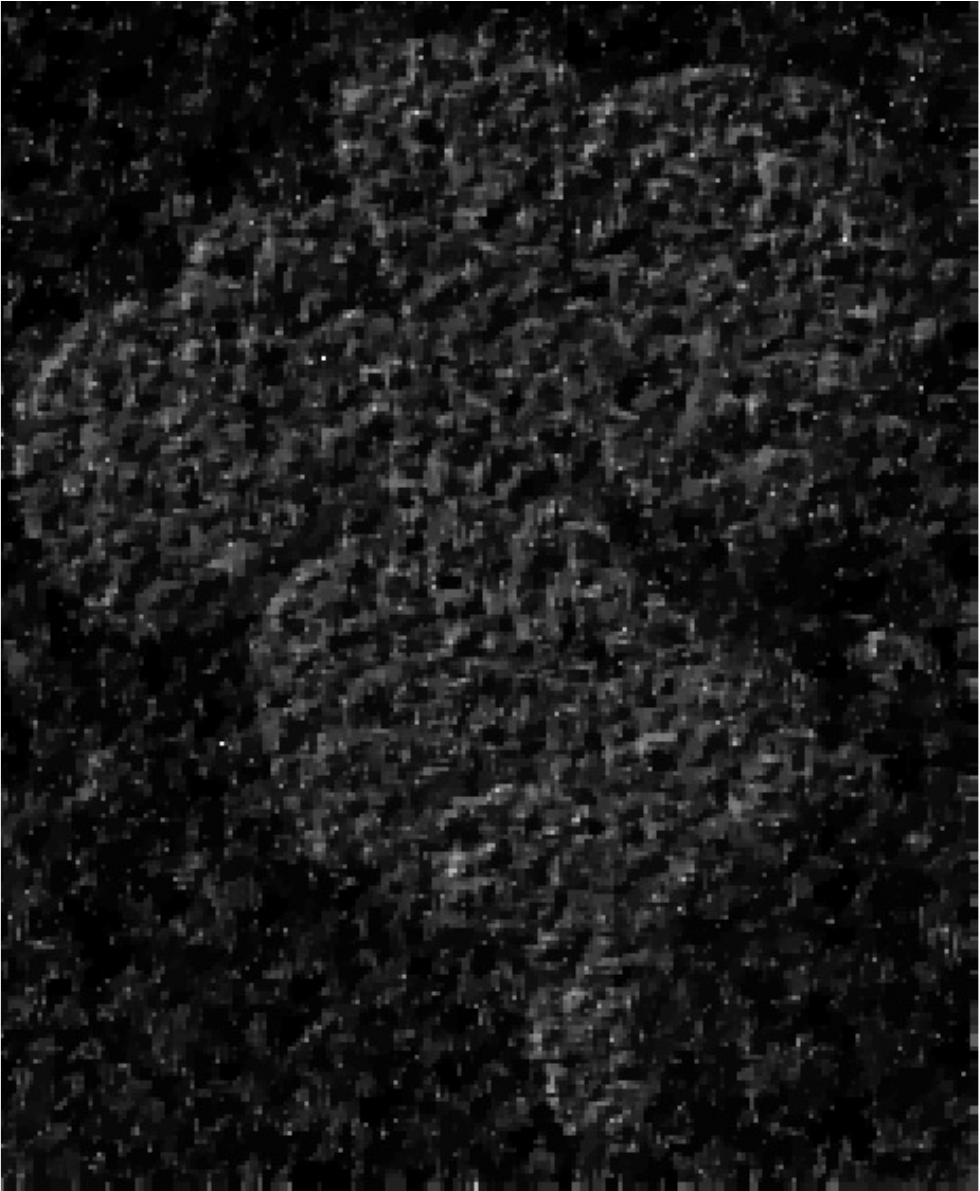} \\
(e)~T-ROF-$\mathcal{S}$ & (f)~ROF-$\mathcal{S}$\\
 \includegraphics[width = 3cm]{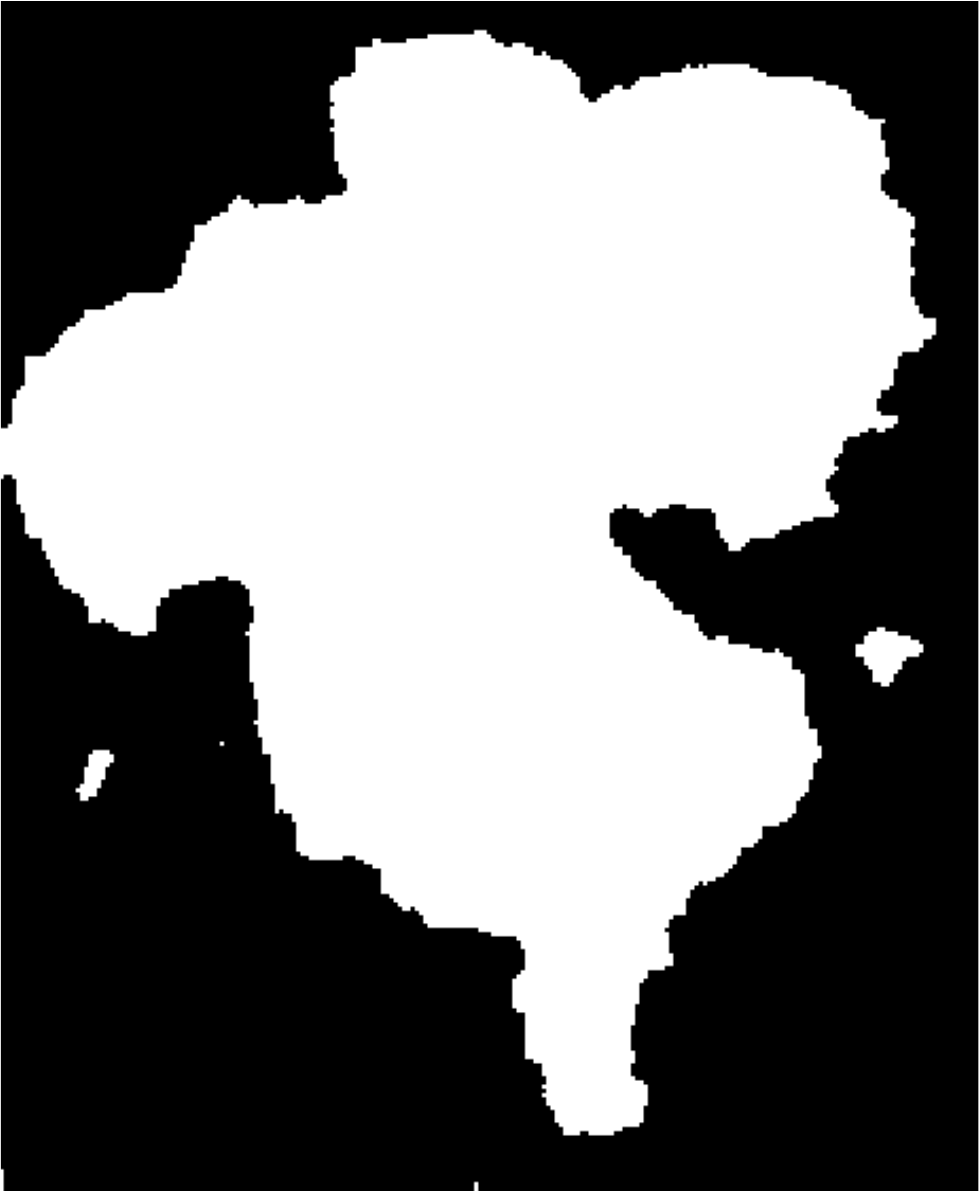} & \includegraphics[width = 3cm]{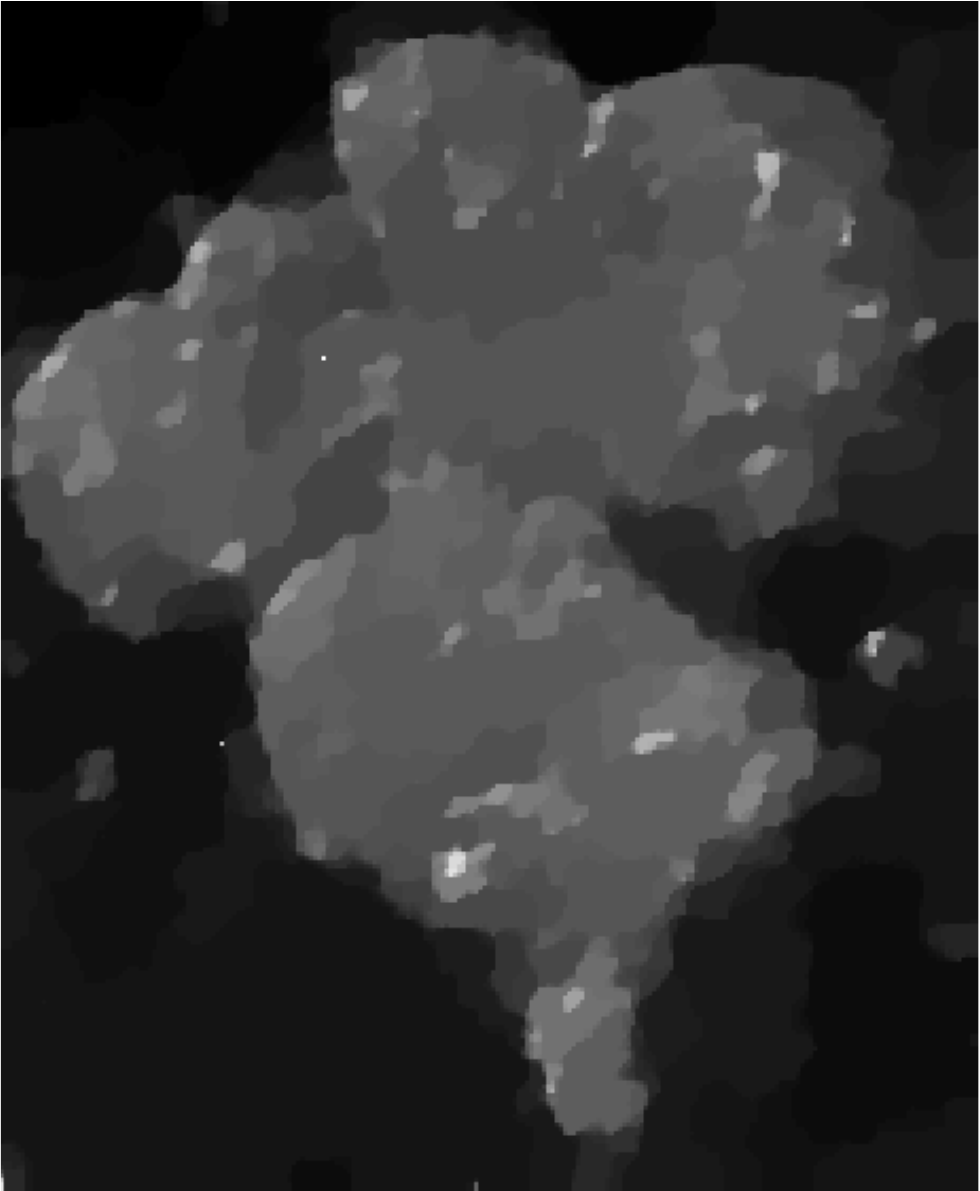} \\
(g)~T-\textit{Joint} & (h)~\textit{Joint} \\
\includegraphics[width = 3cm]{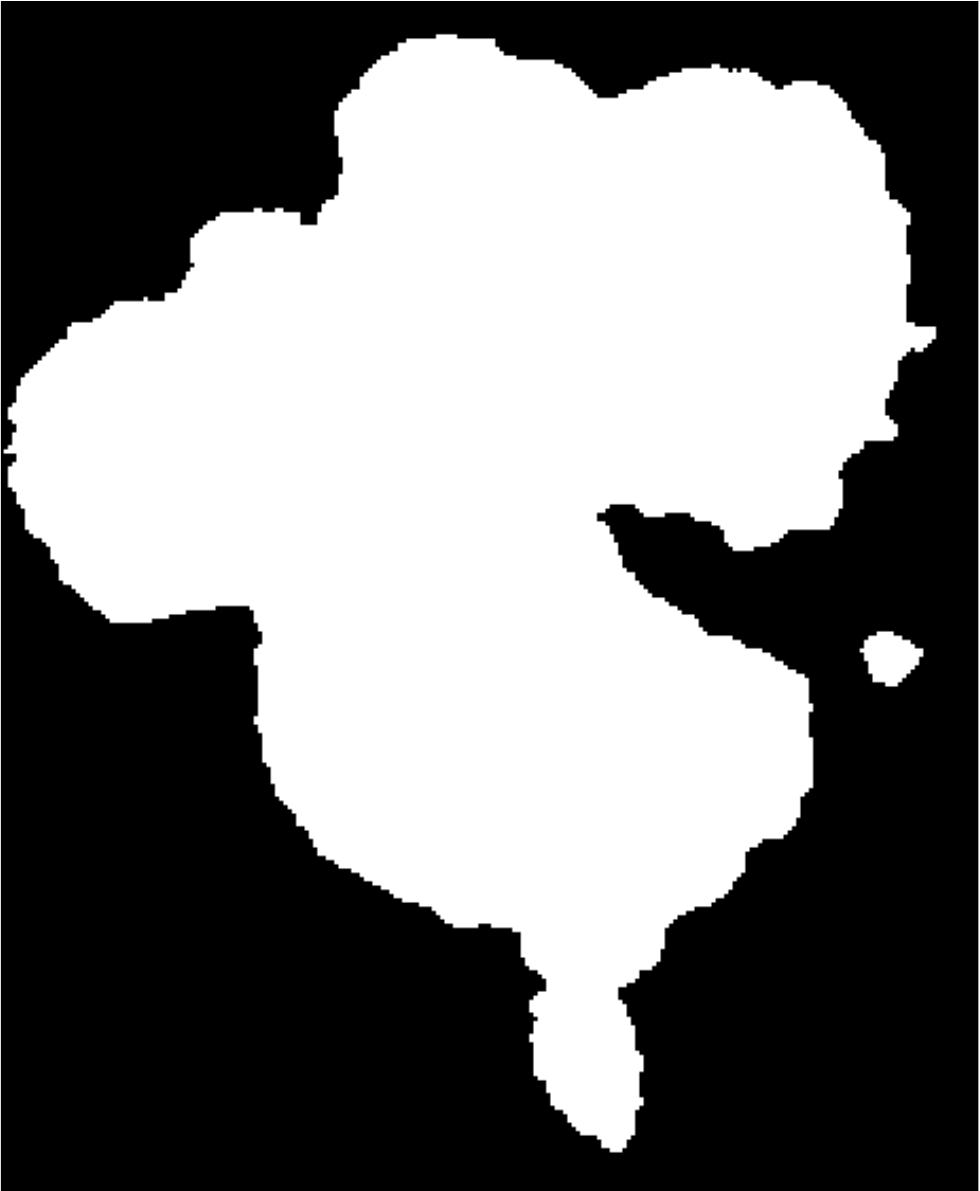} & \includegraphics[width = 3cm]{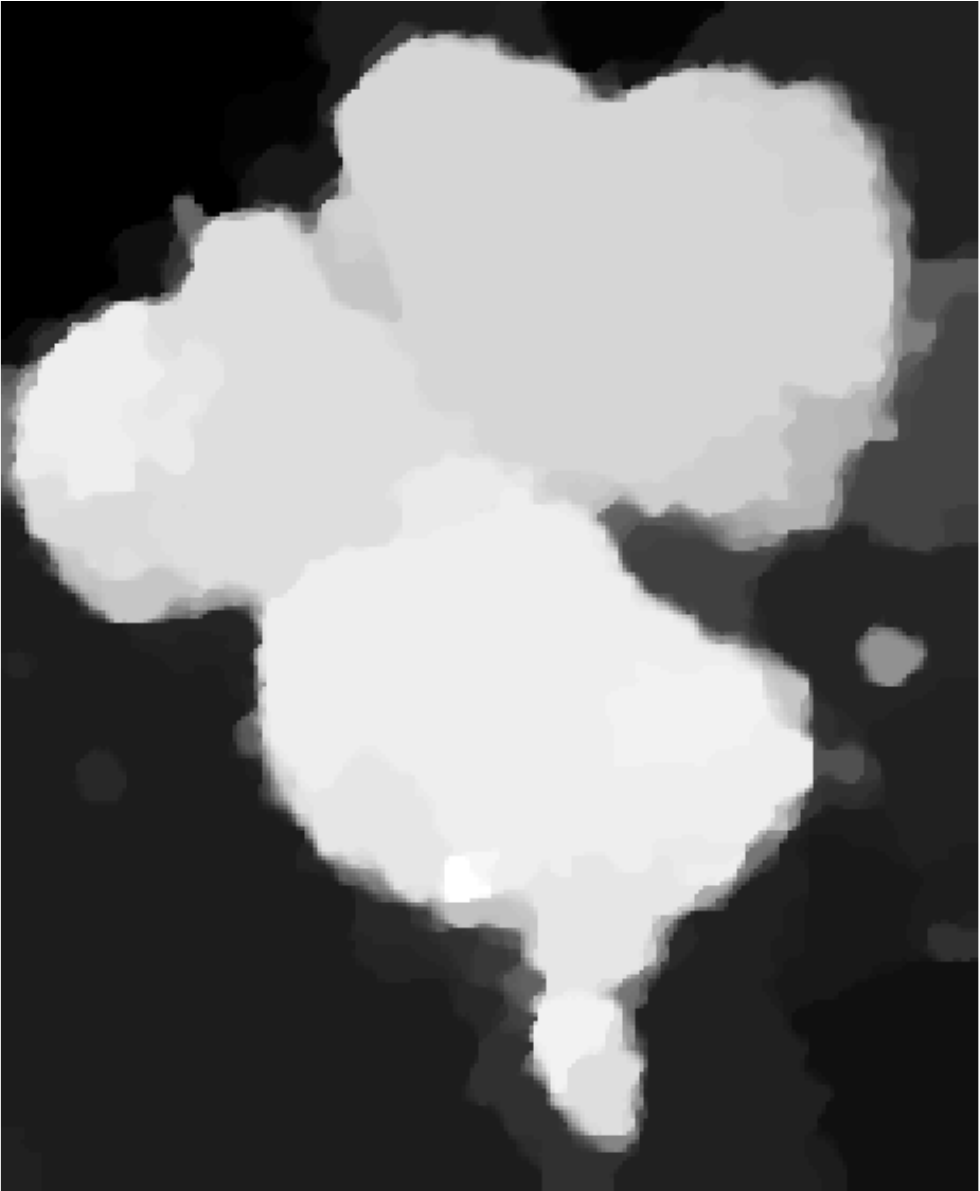} \\
(i)~T-\textit{Coupled} & (j)~\textit{Coupled} \\
\includegraphics[width = 3cm]{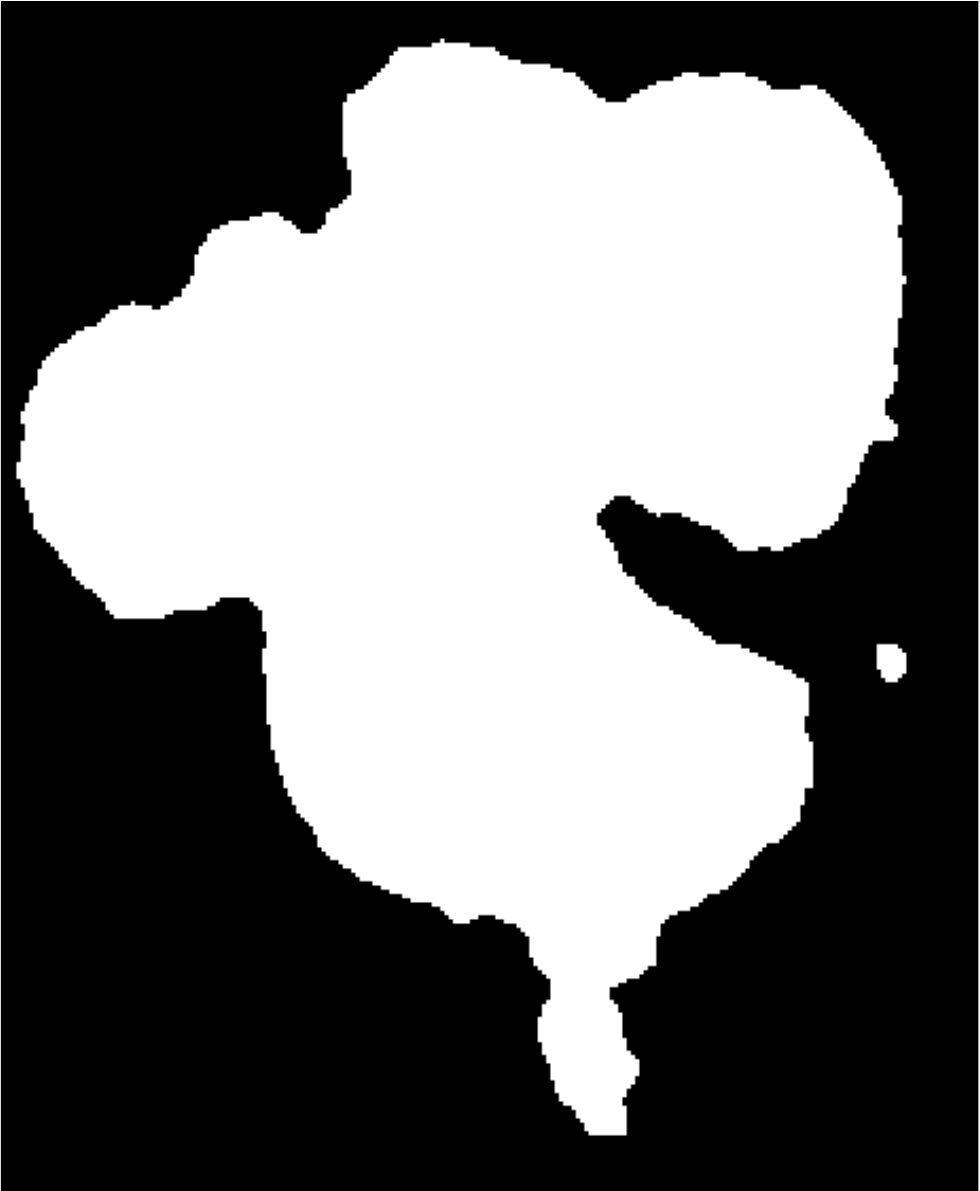} & \includegraphics[width = 3cm]{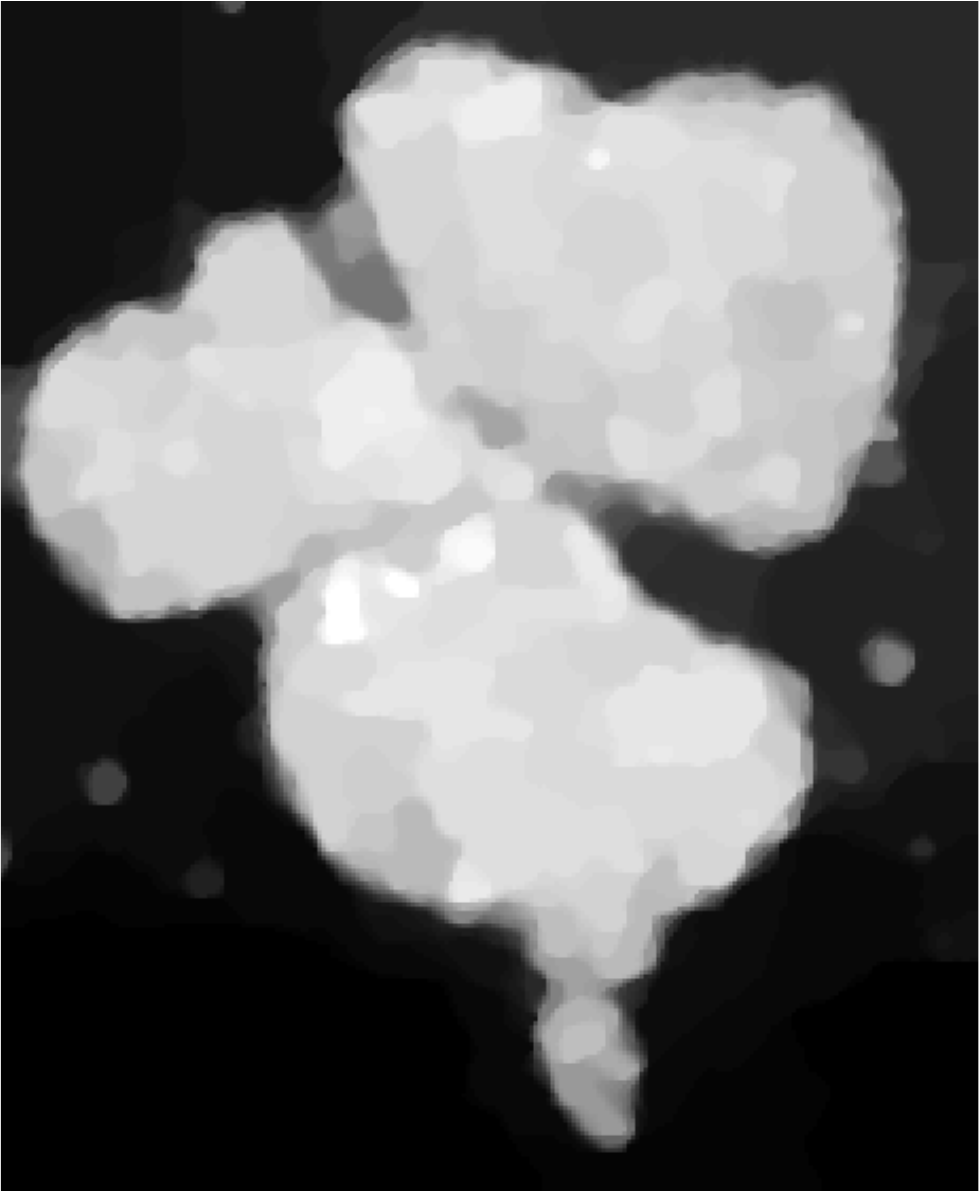}
\end{tabular}
\caption{\label{fig:txt_seg_zoom} {\bf  Porous media multiphase flow texture segmentation based on fractal features.} Comparisons between different approaches summarized in Table~\ref{tab:meths}. \textit{Zoom}  on the area marked by the black rectangle in Figure~\ref{fig:txt_seg}(a).}
\end{figure}

\noindent \textbf{Iterated thresholding} -- From the regularized estimates, $\widehat{h}_{\mathrm{J/C}}$, one can obtain a segmentation by applying a post-processing thresholding.
The iterated thresholding procedure, proposed in~\cite{cai2013multiclass,cai2018linkage}, benefiting from theoretical assessment, is customized to the gas/liquid segmentation problem in Algorithm~\ref{alg:trof}.
It is used systematically in the following, leading to the proposed T-\textit{Joint} and T-\textit{Coupled} segmentation procedures introduced in~\cite{pascal2019joint}.\\ 

\begin{algorithm}
\begin{algorithmic}
\REQUIRE  $\widehat{h} $\\
\ENSURE $\mathrm{m}_0^{[0]} = \underset{\underline{n} \in \Omega}{\min} \, \widehat{h}_{\underline{n}}$, \, 
$\mathrm{m}_1^{[0]} = \underset{\underline{n} \in \Omega}{\max} \, \widehat{h}_{\underline{n}}$.
\FOR{$t \in \mathbb{N}^*$}
\STATE \COMMENT{Compute the threshold:}
\STATE $ \mathrm{T}^{[t-1]} = \left( \mathrm{m}_0^{[t-1]} + \mathrm{m}_1^{[t-1]}\right)/2 $
\STATE \COMMENT{Threshold $\widehat{h}$:}
\STATE $\Omega_0^{[t]} = \lbrace \underline{n} \, | \, \widehat{h}_{\underline{n}} \leq \mathrm{T}^{[t]} \rbrace$, 
\, $\Omega_1^{[t]} = \lbrace \underline{n} \, | \, \widehat{h}_{\underline{n}} > \mathrm{T}^{[t]} \rbrace$ 
\STATE \COMMENT{Update region mean:}
\STATE $\displaystyle \mathrm{m}_0^{[t]} = 1/ \lvert \Omega_0 \rvert \sum_{\underline{n} \in \Omega_0}\widehat{h}_{\underline{n}}$, \,
$\displaystyle \mathrm{m}_1^{[t]} = 1/ \lvert \Omega_1 \rvert \sum_{\underline{n} \in \Omega_1}\widehat{h}_{\underline{n}}$.
\ENDFOR
\RETURN $\Omega_0 = \Omega_0^{[\infty]}$ (liquid), \, $\Omega_1 = \Omega_1^{[\infty]}$ (gas)
\end{algorithmic}
\caption{\label{alg:trof}T-ROF: iterative thresholding of $\widehat{\boldsymbol{h}}_{\mathrm{ROF}}$}
\end{algorithm}

\noindent \textbf{Compared texture segmentation procedures} --
Four procedures falling under Model~\eqref{eq:nl}  and satisfying assumptions of Theorem~\ref{thm:CP} and Theorem~\ref{th:sure} will be compared for texture segmentation, summarized in Table~\ref{tab:meths}.
Note that they differ both by the functional minimized and the noise model, which is of crucial importance in Stein procedures.

The first one, denoted ROF-$\mathrm{Id}$, is a state-of-the-art piecewise constant denoising method, applied on $\widehat{h}_{\mathrm{LR}}$ seen as an observation of $\bar{h}$ corrupted by additive i.i.d. zero-mean Gaussian noise of variance $\sigma^2$, hence with scalar covariance matrix $\mathcal{S} = \sigma^2\mathrm{Id}$.\\
The three procedures ROF-$\mathcal{S}$, \textit{Joint} and \textit{Coupled} take into account the covariance structure of the log-\textit{leaders} coefficients, evidencing both inter-scale and spatial correlations encapsulated in a non-diagonal covariance matrix $\mathcal{S}$.\\
The linear operator intervening in the data fidelity term of \textit{Joint} and \textit{Coupled} procedures, denoted $J$, acts on the double variable $(h,v)$ as $J(h,v) := \left( jh+v\right)_{j = J_1}^{j_2}$.
We showed in a previous work~\cite{pascal2019joint} that it is full-rank.
Hence, Theorem~\ref{thm:CP} applies.
Moreover, the strong-convexity modulus $\mu = 2 \min \mathrm{Sp}(J^{\top}J)$, where $\mathrm{Sp}$ denotes the spectrum of a linear operator, only depends on the octave range $\lbrace j_1, \hdots, j_2\rbrace$ and its numerical values are provided for fixed $j_1 = 1$ and varying $j_2$ in Table~\ref{tab:gamma_J}~\cite{pascal2019joint}.
\\
\setlength{\tabcolsep}{1mm}
\begin{table}[h!]
\centering
\resizebox{\linewidth}{!}{\begin{tabular}{ccccccc}
\toprule
Method & Figures & Observation  & Operator  & Variable  & Penalization & Covariance\\
& \ref{fig:txt_seg}, \ref{fig:txt_seg_zoom} &  $z$ & $A$ & $x$ & &  \\
\midrule
ROF-$\mathrm{Id}$ & (c), (d) & $\widehat{h}_{\mathrm{LR}}$ & $\mathrm{Id}$ & $h$ & TV & $\sigma^2 \mathrm{Id}$\\
ROF-$\mathcal{S}$ & (e), (f) & $\widehat{h}_{\mathrm{LR}}$ & $\mathrm{Id}$ & $h$ & TV &  $\mathcal{S}$\\
\textit{Joint} & (g), (h) & $\log_2(\mathcal{L})$ & $J$ & $(h,v)$ &  $\Psi_{\mathrm{J}}$ & $\mathcal{S}$\\
\textit{Coupled} & (i), (j) & $\log_2(\mathcal{L})$ & $J$ & $(h,v)$ & $\Psi_{\mathrm{C}}$ & $\mathcal{S}$\\
\bottomrule
\end{tabular}}
\caption{Four different settings considered in the experiments of local regularity-based texture segmentation with automated choice of hyperparameters procedures. $\widehat{h}_{\mathrm{LR}}$ stands for the minimizer of $\eqref{eq:defF}$, $\mathcal{L}$ denotes the wavelet leaders of the image to analyze, TV stands for total-variation penalization as defined in \eqref{eq:l12_2D}, and with $h$, $v$, $\Psi_{\mathrm{J}}$ and $\Psi_{\mathrm{C}}$, $\mathcal{S}$ are defined in this section denote respectively the local regularity, the local variance, the Joint penalization, the Coupled penalization and the covariance matrix. \label{tab:meths}}
\end{table}

\noindent \textbf{Automated hyperparameter tuning} -- Stein based formalism, described in Section~\ref{subsec:auto_lambda}, is used, first, to obtain an estimation of the quadratic risk from $\mathrm{SURE}_{\varepsilon, \delta}(\lambda, \alpha)$, second, for automated tuning of regularization parameters thanks to $\mathrm{SUGAR_{\varepsilon, \delta}}(\lambda, \alpha)$ estimate.\\
For this purpose, it is necessary to provide an estimate of the covariance matrix of the noise.
The estimated noise variance $\sigma^2$ involved in ROF-$\mathrm{Id}$ is obtained from the variance of $\widehat{h}_{\mathrm{LR}}$, while the covariance matrix $\mathcal{S}$ is assimilated to the covariance of the log-\textit{leaders} of the textured image $X$ to be segmented.\\
The Finite Difference step $\varepsilon$, involved in $\mathrm{SURE}_{\varepsilon, \delta}$ and $\mathrm{SUGAR}_{\varepsilon, \delta}$ computation (see Equations~\eqref{eq:SURE}~and~\eqref{eq:SUGAR}) is set to
\begin{align}
\varepsilon = \frac{2\sqrt{\max \mathcal{S}}}{M^{0.3}}
\end{align}
where $M$ is the size of the observation vector and the maximum is taken over all coefficients of the covariance matrix and  
$M = N_1 \times N_2$ in the case of ROF-$\mathrm{Id}$ and ROF-$\mathcal{S}$, $M = (j_2 - j_1 +1)\times N_1\times N_2$ in the case of \textit{Joint} and \textit{Coupled} procedures.\\

\setlength{\tabcolsep}{2.5mm}
\begin{table}[h!]
\centering
\begin{tabular}{lccccc}
\toprule
 & $j_2=2$ & $j_2=3$ & $j_2=4$ & $j_2=5$ & $j_2=6$ \\
 \midrule
$\mu$ & $\boldsymbol{0.29}$ & $0.72$ & $1.20$ & $1.69$ & $2.20$ \\
\bottomrule
\end{tabular}
\caption{\label{tab:gamma_J} Strong-convexity modulus $\mu$ of data-fidelity term of~\eqref{eq:nl}, for fixed $j_1=1$ and varied $j_2$. The bold entry correspond to the range of scales used in the experiments.}
\end{table}

\noindent \textbf{Accuracy of the automated tuning} -- 
Grid search minimization of $\mathrm{SURE}_{\varepsilon, \delta}(\lambda,\alpha)$ (Algorithm~\ref{alg:grid}) being costly, due to the large number of Algorithm~\ref{algo:pd_sure} runs required, it is performed on a zoomed image of $281 \times 231$ pixels, presented in Figure~\ref{fig:txt_seg_zoom}(a).
Then, automated tuning of $\lambda$ and $\alpha$ from Algorithm~\ref{alg:BFGS}, based on $\mathrm{SUGAR}_{\varepsilon, \delta}(\lambda, \alpha)$, is performed on the same zoomed image.\\
In practice, $\mathrm{SURE}_{\varepsilon, \delta}(\lambda)$ is computed on 15 values of the hyperparameter $\lambda$ for ROF-$\mathrm{Id}$ (Figure~\ref{fig:SURE}(a)) and ROF-$\mathcal{S}$ (Figure~\ref{fig:SURE}(b)) procedures, and over a $15\times 15$ grid of hyperparameters $(\lambda,\alpha)$ for \textit{Joint} (Figure~\ref{fig:SURE}(c)) and \textit{Coupled} (Figure~\ref{fig:SURE}(d)) methods.
The grid search minimum, $\Lambda_{\mathrm{grid}}$, indicated by the `+' symbol, is compared to the optimal regularization parameters found applying Algorithm~\ref{alg:BFGS}, $\Lambda_{\mathrm{BFGS}}$, indicated by the `$\color{bclair} \boldsymbol{\ast}$' symbol.
The optimal parameters $\Lambda_{\mathrm{grid}}$ and $\Lambda_{\mathrm{BFGS}}$ appear to coincide perfectly for ROF-$\mathrm{Id}$ and \textit{Joint} procedures.
As for ROF-$\mathcal{S}$ and \textit{Coupled} strategies, even though they are different, they are consistent with $\mathrm{SURE}_{\varepsilon, \delta}$ profile, in the sense that they correspond to similar values of $\mathrm{SURE}_{\varepsilon, \delta}$.
We observed that, while grid search minimization (Algorithm~\ref{alg:grid}) required $225$ runs of Algorithm~\ref{algo:pd_sure} for \textit{Joint} and \textit{Coupled} methods, the automated tuning via BFGS quasi-Newton minimization (Algorithm~\ref{alg:BFGS}) needed no more than 50 runs of Algorithm~\ref{algo:pd_sugar}.
Hence, when several parameters are involved, an automated strategy (Algorithm~\ref{alg:BFGS}) is significantly faster than a grid search (Algorithm~\ref{alg:grid}).\\

\begin{figure}[h!]
\centering
\begin{tabular}{rr}
\multicolumn{1}{c}{(a)~ROF-$\mathrm{Id}$} & \multicolumn{1}{c}{(b)~ROF-$\mathcal{S}$}  \\
 \includegraphics[width = 4.5cm]{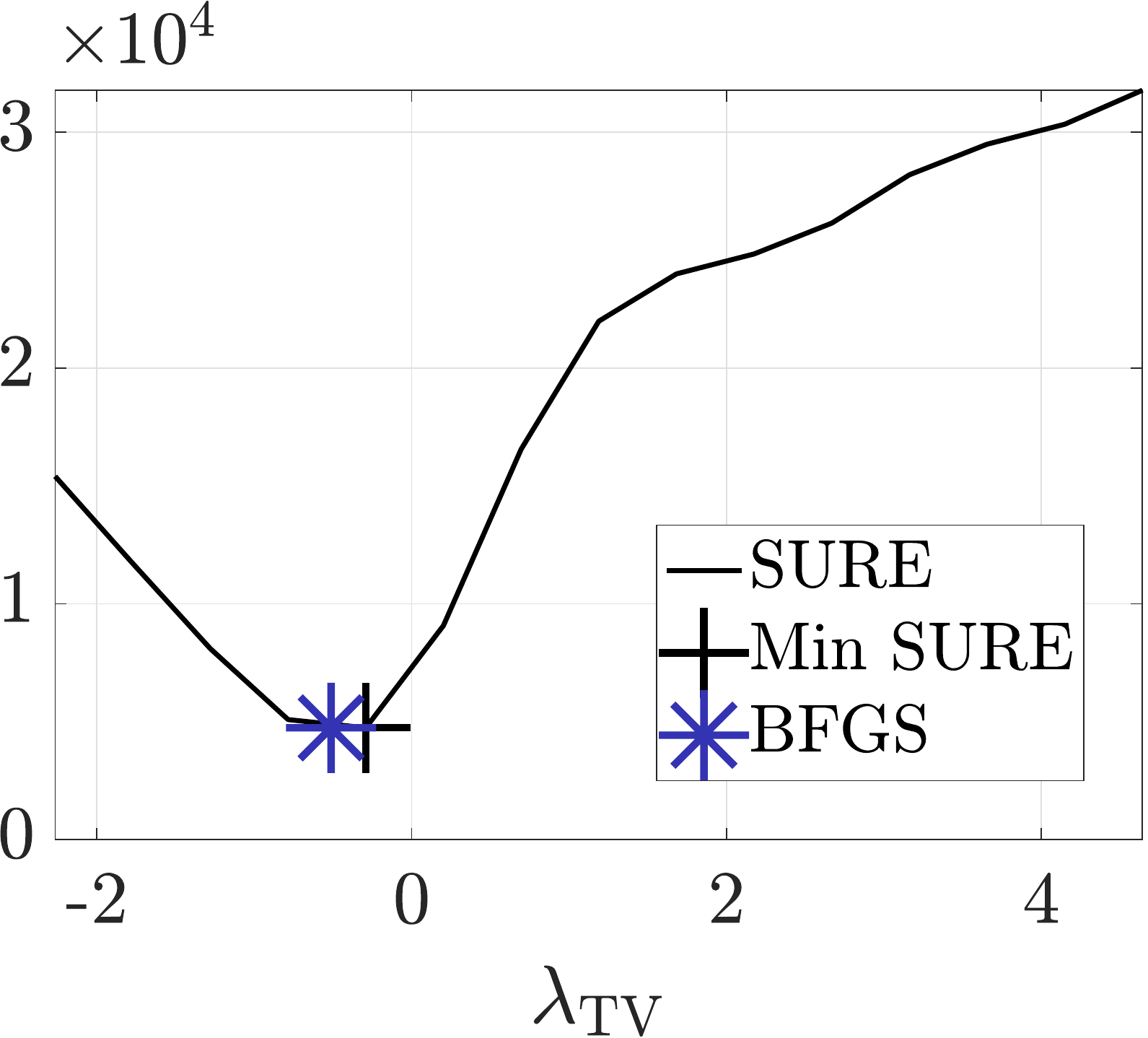} & \includegraphics[width = 4.5cm]{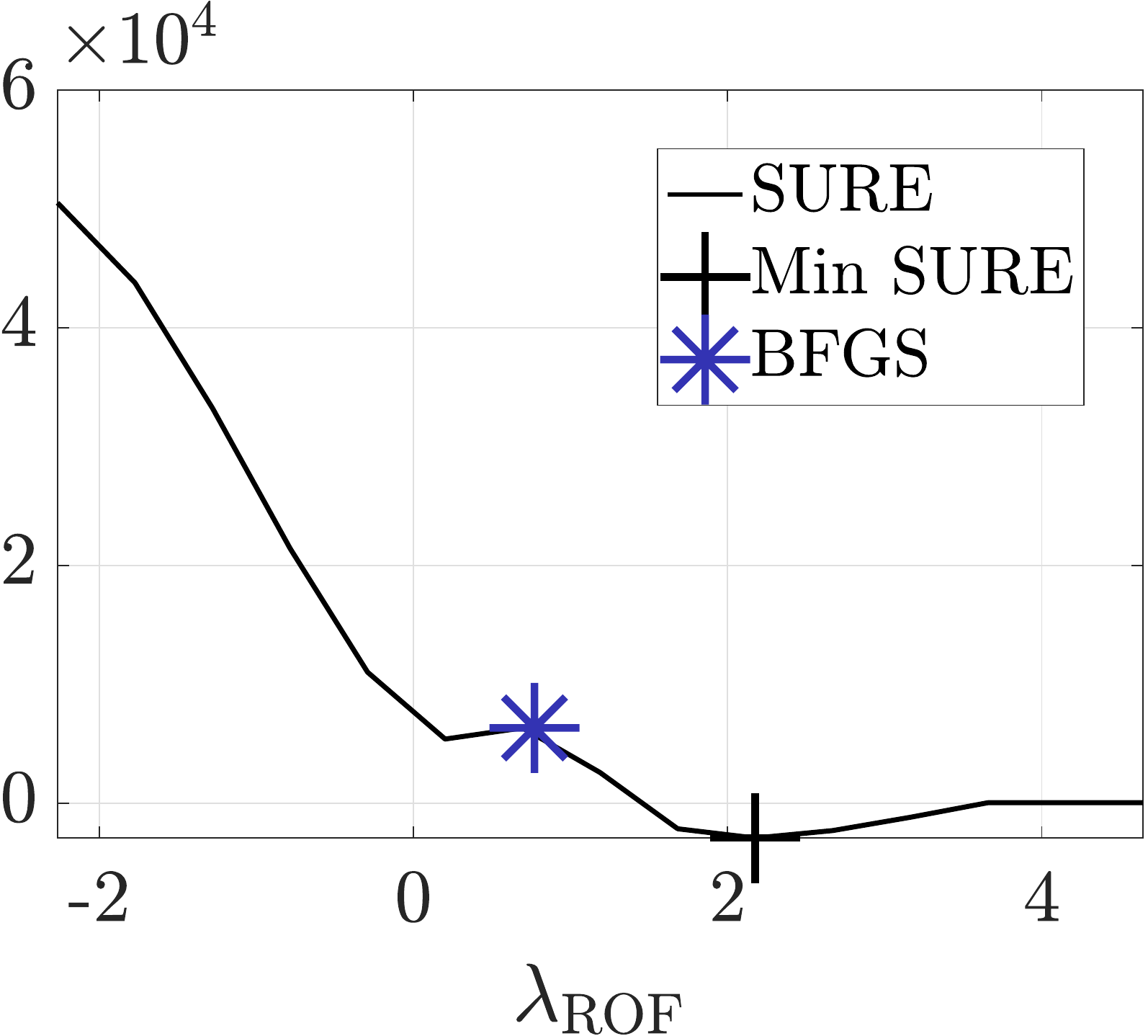} \\
 \noalign{\vspace{3mm}}
 \multicolumn{1}{c}{(c)~Joint}  & \multicolumn{1}{c}{(d)~Coupled}  \\
\includegraphics[width = 5cm]{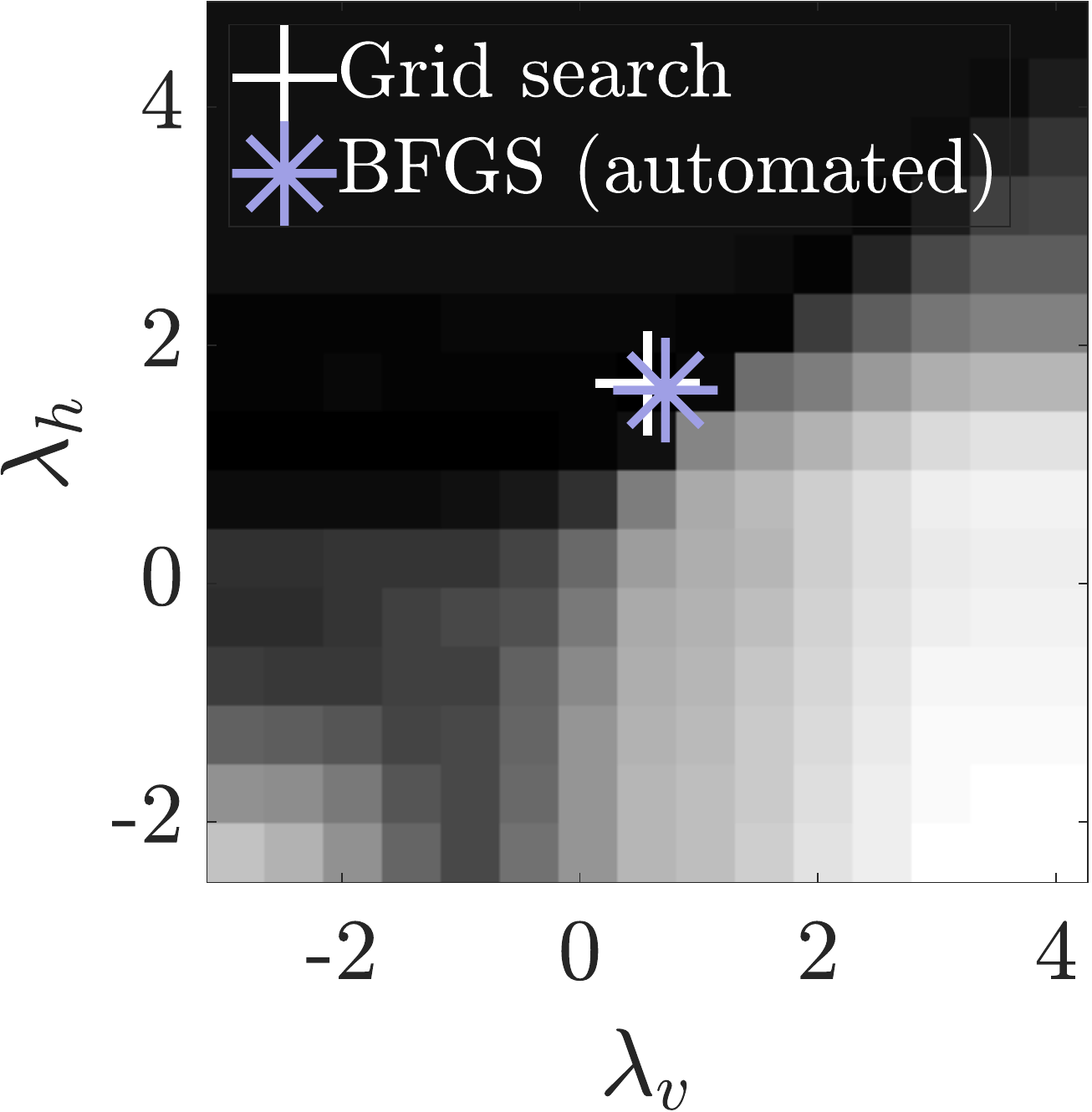} &\includegraphics[width = 5cm]{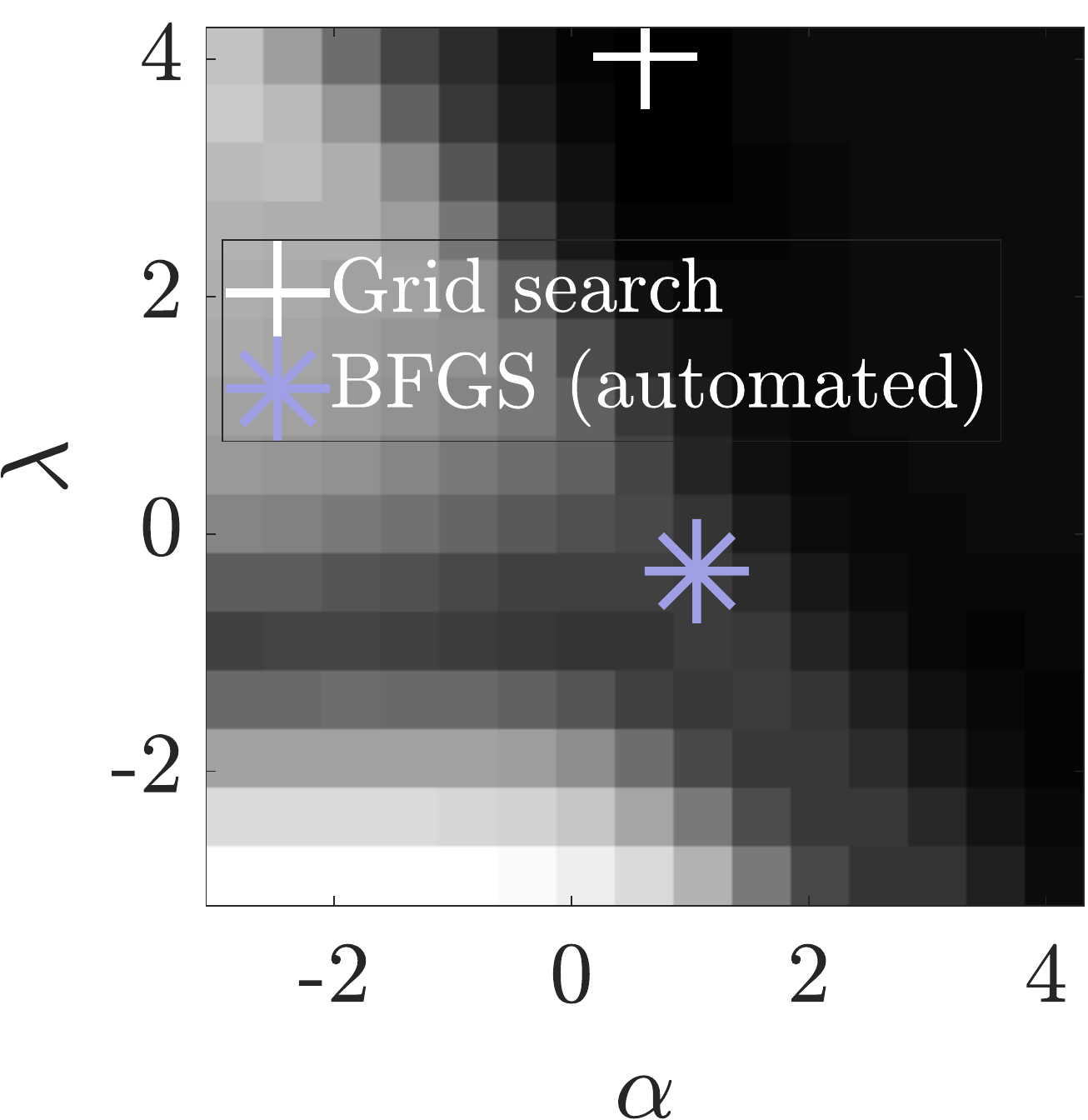} 
\end{tabular}
\caption{Grid search stategy to minimize SURE for the segmentation of a \textit{zoomed} multiphase flow image. \label{fig:SURE}}
\end{figure}

\noindent \textbf{Segmentation results} -- 
Figure~\ref{fig:grid_bfgs_plots} indicates that the automated selection of regularization parameters is consistent with $\mathrm{SURE}_{\varepsilon,\delta}$ minimization.
Hence, the complete images of $1626\times1160$ pixels will be analyzed only with Algorithm~\ref{alg:BFGS}.
The corresponding segmentation results are presented in Figures~\ref{fig:txt_seg}~and~\ref{fig:txt_seg_zoom}. 
State-of-the-art ROF-$\mathrm{Id}$ and ROF-$\mathcal{S}$ procedures yield regularized $\widehat{h}_{\mathrm{TV/ROF}}$ presenting artifacts, as observed in Figures~\ref{fig:txt_seg}(d)~and~\ref{fig:txt_seg}(f), and hence lead to inaccurate segmentation, cf. Figures~\ref{fig:txt_seg}(c)~and~\ref{fig:txt_seg}(e).
In addition, a key point in such experiments is to estimate precisely the contact surface between the liquid and gas. Both T-ROF-$\mathrm{Id}$ and T-ROF-$\mathcal{S}$ (Figures~\ref{fig:txt_seg}(c)~and~(e)) present irregular contours, which are not representative of the real contours and strongly overestimate bubble perimeters. 
\textit{Joint} and \textit{Coupled} procedures, taking into account both the local regularity and the local variance yield more regular contours.
In addition, the \textit{Joint} and \textit{Coupled} methods detect less artifacts (see Figures~\ref{fig:txt_seg}(g)~and~\ref{fig:txt_seg}(i)).
However, the \textit{Joint} estimate $\widehat{h}_{\mathrm{J}}$ appears to be over-regularized, leading to non-detection of small bubbles in the segmentation of Figure~\ref{fig:txt_seg}(g).
The \textit{Coupled} procedure turns out to perform a satisfactory compromise, avoiding artifacts, yet, detecting small gas bubbles, as illustrated in Figures~\ref{fig:txt_seg}(i)~and~\ref{fig:txt_seg}(j).

\section{Conclusion and perspectives} 

The present work has described a unified framework for signal/image non linear filtering, formulated as an inverse problem, that can actually be affiliated to several functional minimization problems encountered in statistical (nonlinear) physics. 
This inverse problem formulation aims at favoring piecewise homogeneous signal and images, that naturally correspond to solutions on numerous problems in nonlinear physics, often very different in nature. 
Piecewise homogeneity assessment entails non smooth convex optimization, here handled via proximal operators. 
In addition to yielding relevant piecewise homogeneous estimates, the proposed framework also achieves an automated and data-driven tuning of  hyperparameters inherently present in inverse problems and nonlinear filtering, thus avoiding the burden of conducting a prone to error and sometimes lacking reproductibily expert inspection.
The potential and interest of nonlinear filtering has been illustrated at work on two, different in nature, real nonlinear physics experiments (low confinement solid friction and porous media multiphase flow). 
However, the approach has a fairly general level of applicability and a documented {\sc Matlab} toolbox both for multivariate signals and images, implementing both the nonlinear filtering favoring piecewise homogeneity and the automated data-driven hyperparameter selection, has been made publicly available at \texttt{https://github.com/bpascal-fr/stein-piecewise-filtering}.

\bibliographystyle{plain}
\bibliography{abbr,biblio}

\end{document}